\documentclass[preprint,10pt,3p,authoryear]{elsarticle}

\usepackage{latexsym}
\usepackage{amsmath, amsfonts, amsthm,amssymb}
\usepackage{epsfig}
\usepackage{stmaryrd}
\usepackage{multicol}
\usepackage{pdfsync}
\usepackage{bm}
\usepackage{dsfont}
\usepackage{comment}
\usepackage{algorithmic}
\usepackage{algorithm}
\usepackage{caption}
\usepackage{psfrag}
\usepackage{subfig}
\usepackage{xcolor}
\usepackage{graphics}
\usepackage{enumitem}
\usepackage{mathtools}
\usepackage{bbm}
\usepackage{mathrsfs} 
\usepackage{color}
\usepackage{hyperref}
\usepackage{cancel}
\usepackage{float}
\usepackage{todonotes}
\usepackage{setspace}

\usepackage[utf8]{inputenc}
\usepackage{titlesec}
\titleformat{\subsection}
{\normalfont\fontsize{10}{17}\bfseries}{\thesubsection}{1em}{}

\usepackage{tikz-cd}
\journal{TBA}
\hypersetup{
	colorlinks = true, 
	urlcolor = blue, 
	linkcolor = red, 
	citecolor = blue 
}

\newtheorem{theorem}{Theorem}
\newtheorem{rem}{Remark}
\newtheorem{cond}{Condition}

\newtheorem{lemma}{Lemma}
\newtheorem{proposition}{Proposition}

\theoremstyle{remark}

\numberwithin{equation}{section}

\DeclareMathOperator*{\esssup}{ess\,sup}

\mathtoolsset{showonlyrefs=true}

\title{
\textbf{
Equilibrium Reward for Liquidity Providers in Automated Market Makers
}
}

\author[label1,label2]{Alif Aqsha}

\address[label1]{Mathematical Institute, University of Oxford}
\address[label2]{Oxford-Man Institute of Quantitative Finance}
\address[label3]{Ceremade, Université Paris Dauphine-PSL}
\ead{alif.aqsha@maths.ox.ac.uk}
\author[label3]{Philippe Bergault}
\ead{bergault@ceremade.dauphine.fr}
 \author[label1,label2]{Leandro S\'{a}nchez-Betancourt}
\ead{sanchezbetan@maths.ox.ac.uk}

\DeclareMathSymbol{\shortminus}{\mathbin}{AMSa}{"39}

\begin{document}

\newcommand{\la}{\left \langle}
\newcommand{\ra}{\right\rangle}
\newcommand{\cb}[1]{{\color{blue} #1}}
\newcommand{\norm}[1]{\left\lVert #1 \right\rVert}
\newcommand{\bae}{\begin{equation}\begin{aligned}}
\newcommand{\eae}{\end{aligned}\end{equation}}
\newcommand{\beq}{\begin{equation}}
\newcommand{\eeq}{\end{equation}}
\newcommand{\N}{\mathbb{N}}
\newcommand{\R}{\mathbb{R}}
\newcommand{\E}{\mathbb{E}}
\newcommand{\Pb}{\mathbb{P}}
\newcommand{\PbI}{\mathbb{P}^I}
\newcommand{\PbB}{\mathbb{P}^B}
\newcommand{\Lb}{\mathbb{L}}

\newcommand{\mfT}{{\mathfrak{T}}}
\newcommand{\mfO}{{\mathfrak{O}}}
\newcommand{\tT}{{t\in\mfT}}
\newcommand{\mcA}{{\mathcal{A}}}
\newcommand{\mcC}{{\mathcal{C}}}
\newcommand{\mcF}{{\mathcal{F}}}
\newcommand{\mcN}{{\mathcal{N}}}
\newcommand{\mcB}{{\mathcal{B}}}
\newcommand{\mcH}{{\mathcal{H}}}

\newcommand{\transB}{\mathfrak{t}}
\newcommand{\decayB}{\mathfrak{p}}
\newcommand{\instantB}{\mathfrak{h}}

\newcommand{\assign}{:=}
\newcommand{\nobracket}{}
\newcommand{\tmop}[1]{\ensuremath{\operatorname{#1}}}
\newcommand{\tmtextit}[1]{\text{{\itshape{#1}}}}
%

\newcommand{\F}{\mathcal{F}}
\newcommand{\Fb}{\mathbb{F}}
\newcommand{\Prob}{\mathbb{P}}
\newcommand{\X}{\mathbb{X}}
\newcommand{\Ss}{\mathcal{S}}

\newcommand{\Real}{\mathbb{R}}
\newcommand{\Y}{\mathbb{\mathcal{Y}}}
\newcommand{\tmL}{\mathbb{L}}
\newcommand{\Ll}{\mathcal{L}}
\newcommand{\Pp}{\mathcal{P}}
\newcommand{\Bb}{\mathcal{B}}
\newcommand{\Dd}{\mathcal{D}}
\newcommand{\Ee}{\mathcal{E}}
\newcommand{\Nn}{\mathcal{N}}
\newcommand{\Q}{\mathbb{Q}}
\newcommand{\ch}{\mathds{1}}

\newcommand{\wealth}{\text{M}}
\newcommand{\rwealth}{Q}
\newcommand{\rwealthD}{q}

\newcommand{\tempLP}{\mathfrak{a}}
\newcommand{\tempB}{\mathfrak{a}}
\newcommand{\tempI}{\mathfrak{b}}
\newcommand{\tempU}{\mathfrak{c}}
\newcommand{\permB}{\mathfrak{p}}
\newcommand{\termpB}{\phi}
\newcommand{\termpI}{\psi}
\newcommand{\runnB}{r^B}
\newcommand{\runnI}{r^I}

\renewcommand{\d}{\mathrm{d}}

\newcommand{\stateB}{\mathfrak{y}}
\newcommand{\stateI}{\mathfrak{x}}
\newcommand{\state}{\mathfrak{y}}
\newcommand{\statevn}{\mathfrak{y}}

\newcommand{\nuI}{\eta}
\newcommand{\nuB}{\nu}
\newcommand{\nuU}{\xi}

\newcommand{\nuBstar}{\nu^*}
\newcommand{\nuIstar}{\eta^*}

\newcommand{\varI}{\mathbb{V}^I}
\newcommand{\varB}[1][B]{\mathbb{V}^{#1}}

\newcommand{\betaI}{\beta_0^I}
\newcommand{\betaaI}{\beta_1^I}
\newcommand{\rhoI}{\rho_0^I}
\newcommand{\rhooI}{\rho_1^I}

\newcommand{\betaB}{\beta_0^B}
\newcommand{\betaaB}{\beta_1^B}
\newcommand{\rhoB}{\rho_0^B}
\newcommand{\rhooB}{\rho_1^B}

\newcommand{\runcostI}{\rho_0^I + \rho_1^I\, \varI}
\newcommand{\runcostB}{\rho_0^B + \rho_1^B\, \varB}

\newcommand{\termcostI}{\beta_0^I + \beta_1^I\, \varI}
\newcommand{\termcostB}{\beta_0^B + \beta_1^B\, \varB}

\newcommand{\constrB}{\tempB - f_2^2\,\tempI}

\newcommand{\MI}[1][I]{M^{#1}}
\newcommand{\MB}[1][B]{M^{#1}}
\newcommand{\MBa}[1][B]{\tilde{N}^{#1}}
\newcommand{\MBb}[1][B]{\tilde{M}^{#1}}
\newcommand{\MZ}[1][Z]{M^{#1}}

\newcommand{\QI}[1][I]{Q^{#1}}
\newcommand{\QB}[1][B]{Q^{#1}}
\newcommand{\QU}[1][U]{Q^{#1}}
\newcommand{\XI}[1][I]{X^{#1}}
\newcommand{\XB}[1][B]{X^{#1}}
\newcommand{\XU}[1][U]{X^{#1}}

\newcommand{\QIstar}[1][I]{Q^{#1,*}}
\newcommand{\QBstar}[1][B]{Q^{#1, *}}
\newcommand{\XIstar}[1][I]{X^{#1, *}}
\newcommand{\XBstar}[1][B]{X^{#1, *}}

\newcommand{\mcFB}{\mcF^B}
\newcommand{\mcFI}{\mcF^I}
\newcommand{\vfB}{J^{B*}}
\newcommand{\vfI}{J^{I*}}
\newcommand{\pcB}{J^B}
\newcommand{\pcI}{J^I}
\newcommand{\mcAB}{\mcA^B}
\newcommand{\mcAI}{\mcA^I}

\newcommand{\g}{g}
\newcommand{\gI}{g^{I}}
\newcommand{\gB}{g^{B}}
\newcommand{\gZ}{g^{Z}}
\newcommand{\gY}{g^{Y}}
\newcommand{\h}{h}
\newcommand{\hI}{h^{I}}
\newcommand{\hB}{h^{B}}
\newcommand{\hZ}{h^{Z}}
\newcommand{\hY}{h^{Y}}
\newcommand{\f}{f}
\newcommand{\fI}{f^{I}}
\newcommand{\fB}{f^{B}}
\newcommand{\fZ}{f^{Z}}
\newcommand{\fY}{f^{Y}}

\newcommand{\Pn}{P^{\hat{\nu}}}
\newcommand{\Pa}{P^{\hat{\alpha}}}

\newcommand{\sigmaU}{\sigma^U}
\newcommand{\sigmaM}{\sigma^M}
\newcommand{\sigmaS}{\sigma^S}

\newcommand{\vp}{\varphi}

\newcommand{\rew}{\mathfrak{R}}
\newcommand{\fee}{\mathfrak{r}}
\newcommand{\feeV}{\mathfrak{r}}
\newcommand{\cost}{\mathfrak{c}}

\newcommand{\Na}{N^-}
\newcommand{\Nb}{N^+}
\newcommand{\Nab}{N^\pm}
\newcommand{\Ni}{N^i}
\newcommand{\hatNa}{\hat{N}^-}
\newcommand{\hatNb}{\hat{N}^+}
\newcommand{\hatNab}{\hat{N}^\pm}
\newcommand{\hatNi}{\hat{N}^i}
\newcommand{\tildeNa}{\tilde{N}^-}
\newcommand{\tildeNb}{\tilde{N}^+}
\newcommand{\tildeNab}{\tilde{N}^\pm}
\newcommand{\tildeNi}{\tilde{N}^i}

\newcommand{\lambdaa}{\lambda^-}
\newcommand{\lambdab}{\lambda^+}
\newcommand{\lambdaab}{\lambda^\pm}
\newcommand{\lambdai}{\lambda^i}
\newcommand{\barlambdaa}{\bar{\lambda}^-}
\newcommand{\barlambdab}{\bar{\lambda}^+}
\newcommand{\barlambdaab}{\bar{\lambda}^\pm}
\newcommand{\barlambdai}{\bar{\lambda}^i}

\newcommand{\Aa}{A^-}
\newcommand{\Ab}{A^+}
\newcommand{\Aab}{A^\pm}
\newcommand{\Ai}{A^i}
\newcommand{\AW}{A^W}
\newcommand{\AB}{A^B}

\newcommand{\Deltarwealtha}{\Delta^-} 
\newcommand{\Deltarwealthb}{\Delta^+}
\newcommand{\Deltarwealthab}{\Delta^\pm}
\newcommand{\Deltarwealthi}{\Delta^i}

\newcommand{\deltaa}{\delta^-}
\newcommand{\deltab}{\delta^+}
\newcommand{\deltai}{\delta^i}
\newcommand{\deltaab}{\delta^\pm}

\newcommand{\Za}{Z^-}
\newcommand{\Zb}{Z^+}
\newcommand{\Zi}{Z^i}
\newcommand{\Zab}{Z^\pm}

\newcommand{\va}{v^-}
\newcommand{\vb}{v^+}
\newcommand{\vi}{v^i}
\newcommand{\vab}{v^\pm}

\newcommand{\Vvn}{V}  
\newcommand{\Vvna}{V^-}
\newcommand{\Vvnb}{V^+}
\newcommand{\Vvni}{V^i}

\setlength\parindent{0pt}

\begin{abstract}
We find the equilibrium contract that  an automated market maker (AMM) offers to their strategic liquidity providers (LPs) in order to maximize the order flow that gets processed by the venue. Our model is formulated as a leader-follower stochastic game, where the venue is the leader and a representative LP is the follower. We derive approximate closed-form equilibrium solutions to the stochastic game and analyze the reward structure. 
Our findings suggest that under the equilibrium contract, LPs have incentives to add liquidity to the pool only when higher liquidity on average attracts more noise trading.
The equilibrium contract depends on the external price, the pool reference price, and the pool reserves. Our framework offers  insights into AMM design for maximizing order flow while ensuring LP profitability. 
\vspace{0.5cm}

\end{abstract}

\maketitle

\section{Introduction}

Automated market makers (AMMs) are one of the latest developments in financial technology. An AMM is a venue with predefined trading rules where liquidity takers (LTs) and liquidity providers (LPs) trade. These venues became popular after Uniswap V2 was released in May 2020. As of 2025, the fourth version of the protocol has already been announced; see \cite{adams2021uniswap,adams2023uniswap}.\\

The academic literature on AMMs is still young but already quite rich. The first papers studying the mechanisms of AMMs are in \cite{chiu2019blockchain,angeris2020improved, angeris2021analysis, capponi2021adoption, lipton2021blockchain}, the arbitrage between AMMs and external centralized venues is studied in \cite{cartea2023decentralised,cartea2023execution}, and price formation in AMMs is studied in \cite{capponi2024price}. It is now well-known that, in the absence of fees, an agent providing liquidity in a standard constant function market (CFM) ---where prices are only determined as a function of the reserves--- are exposed to a concave payoff that is inferior to that of holding coins outside of the
pool, a phenomenon called impermanent loss (IL). In fact, empirical studies such as \cite{loesch2021impermanent} show that currently, even in the presence of fees, LPs are  on average incurring a loss. Recently, \cite{fukasawa2023model} study the hedging of this impermanent loss and  \cite{milionis2022automated} introduce the notion of loss-versus-rebalancing  
---see also \cite{milionis2023automated}. Lastly,  \cite{cartea2024decentralized}  consider a more general notion of predictable loss and derive optimal liquidity provision strategies.\\ 

In order to tackle the IL issue, another active area of research is that of AMM design. \cite{bergault2024automated} proposes an alternative mechanism in which the pricing function uses external information about the current  market exchange rates, allowing the AMM to update its bid and offer prices not only after a trade but also after a price oracle has fired an update.\footnote{ Models with more complex price dynamics have been considered in \cite{bergault2024price}; recently, \cite{bergault2024pegged} study the dynamics of stablecoins.  } \cite{cartea2024strategic} introduce a new family of AMMs that they call decentralized liquidity pools (DLPs). In their paper they optimize a range of performance criteria involving the profitability of LPs. The optimal policy of a given performance criterion gives rise to a  new DLP design. Depending on the model, the resulting DLP could monitor external prices or even filter arrivals in search of fundamental values.\footnote{\cite{cartea2023automated} is an earlier version of this paper where arithmetic liquidity pools were introduced. } Lastly, \cite{adams2024amm} introduce the so-called `auction-managed AMM', that reduces losses to informed order flow, and maximizes revenue from uninformed order flow. With the exception of \cite{adams2024amm}, these designs do not aim to maximize the trading activity in the venue. \\

In this paper we find the optimal contract  that a CFM should offer to their strategic LPs in order to attract the maximum amount of order flow to the CFM. We formulate the design as the solution to a leader-follower stochastic game in which the venue is the leader and the representative LP is the follower.\footnote{Similar to \cite{FukasawaLiquidity2025}, we study an aggregate of LPs as a representative provider. }  Closest to our paper is \cite{euch2021optimal} who work in
a centralized limit order book (LOB) and find the optimal make-take fees in such a setting;\footnote{In their paper, the regulator is the leader and the LPs are the followers.} see also \cite{baldacci2021optimal} for a multi-player version of the model, and \cite{baldacci2024optimal} for a different approach to make-take fees. In our paper, the entity designing the contract is the CFM. We find approximate closed-form solutions to the stochastic game and obtain insights into the reward structure that attracts, in equilibrium,  the highest possible order flow to the venue. \\

To the best of our knowledge, this is the first article to study the design of AMMs as a leader-follower stochastic game. In our framework, LPs aim to maximize exponential utility of wealth and the venue wishes to maximize trading activity. As customary in leader-follower games, for a given reward structure, we compute the optimal liquidity providing strategy from the representative LP (follower), and using this, we compute the optimal reward structure to be provided by the venue (leader). The new design comes in the form of a ``contract'' that encodes the reward structure obtained in the second optimization. 
We find that if higher liquidity (in the pool) does not attract noise trading, then LPs do not have incentives to add liquidity to the pool.
On the other hand, if higher liquidity attracts noise trading, then, the optimal contract incentivises LPs to add liquidity to the pool in order to benefit from the noise order flow. In such a scenario, the transaction costs in the external market scale the amount that the representative LP adds to the pool.
In equilibrium, the optimal contract depends on the external price, the pool’s reference price, and the reserves in the pool. Our equilibrium solutions provide insights into the design of these emerging venues. 
\\

The remainder of the paper proceeds as follows.
Section \ref{sec:model} introduces the probabilistic framework of the model. In particular, Section \ref{sec:representative-LP-problem} introduces the problem of the representative LP and Section \ref{sec:venue-problem} introduces the problem of the venue.
The solutions to these two problems are derived in Sections \ref{sec:solution-LP-problem} and \ref{sec:solution-venue-problem}, respectively. Finally, Section \ref{sec:numerical-results} presents the numerical results and we collect proofs in the appendix.

\section{The model}\label{sec:model}

\subsection{Probabilistic framework}

Let $T>0$ be a trading horizon (e.g., one day), $\Omega_c$ the set of continuous functions from $\mfT = [0,T]$ into $\mathbb R$, $\Omega_d$ the set of piecewise constant càdlàg functions from $\mfT$ into $\mathbb N$, and $\Omega = \big(\Omega_c\big)^2 \times \big( \Omega_d \big)^2$ with the corresponding Borel algebra $\mathcal F$. The observable state is the canonical process $ (\chi_t)_{t \in \mfT} = \big(W_t, B_t, \hatNb_t, \hatNa_t \big)_{t \in \mfT}$ of the measurable space $(\Omega, \mathcal F)$, with
\[
 W_t(\omega) := w(t), \; B_t(\omega) := b(t), \; \hatNb_t(\omega) :=  \hat n^+(t), \; \hatNa_t(\omega) :=  \hat n^-(t), \; \text{for all } t \in \mfT,
\]
where $\omega := (w, b, \hat n^+, \hat n^-) \in \Omega.$\\

We introduce a probability measure $\Pb$ such that $W$ and $B$ are Brownian motions and $\hatNa$, $\hatNb$ are Poisson processes with intensity $a_0>0$. In this probability measure, $W$, $B$, $\hatNa$ and $\hatNb$ are independent.\footnote{The independence assumption between $W$ and $B$ can be relaxed.  }\\

We study trading in a pair of assets $X$ and $Y$ (e.g., USDC and ETH) that takes place in a constant function market (CFM). In addition to the CFM, there is an external limit order book (LOB) venue where trading in $X$ and $Y$ occurs. We let $S=(S_t)_{t\in \mfT}$ be the external midprice of asset $Y$ in terms of asset $X$. Within the CFM, we let $X_t$ and $Y_t$ be the quantities of assets $X$ and $Y$  in the pool at time $t\in\mfT$. The CFM has trading function $f(x,y)=x\,y$; this is the most popular choice of $f$ and is known as the constant product market (CPM).  We let the function $\varphi_{c_t}(y)$ be the level function of $f$ such that $\varphi_{c_t}(Y_t) = X_t$ with 
$$c_t = f(X_{t\shortminus}, Y_{t\shortminus}) = X_{t\shortminus}\,Y_{t\shortminus}\,.$$ 
More precisely, we have that
\begin{align}
    \varphi_{c}(y) &= \frac{c}{y}\,,\qquad 
    (\varphi_{c})'(y) = -\frac{c}{y^2}\,.
\end{align}

The external midprice satisfies
\begin{equation}
    \d S_t = \sigma \,\d W_t\,,\qquad S_0\in\R^+\,,
\end{equation}
where $\sigma>0$ is the volatility parameter. We think of the external price as that of a venue where price formation takes place (e.g., Binance, Kraken). \\

In the CFM, liquidity takers arrive according to the counting processes $\Na$ (LT buys) and $\Nb$ (LT sells) that model the number of trades of size $\xi>0$ through time. We denote $Z=(Z_t)_{t\in \mfT}$ as the marginal price of $Y$ in terms of $X$ in the venue. We follow the characterization of CFMs in \cite{cartea2024strategic} that describes the mechanics of the reserves $(X_t)_{t\in \mfT}$, $(Y_t)_{t\in \mfT}$, and the instantaneous rate $(Z_t)_{t\in \mfT}$,  according to the arrival of orders $\Na$ and $\Nb$, and extend it to include stochastic pool depths (given by $c_t$) due to the activity of LPs. In our model we write 
\begin{align*}
    \d X_t &= \underbrace{ Z_t\, \eta \, \d B_t }_{\text{activity from LPs}} \underbrace{ - \left[ \vp_{c_t}(Y_{t\shortminus}) - \vp_{c_t}(Y_{t\shortminus} + \xi) \right]\, \d \Nb_t + \left[ \varphi_{c_t}(Y_{t\shortminus} - \xi) - \vp_{c_t}(Y_{t\shortminus}) \right]\, \d \Na_t }_{\text{activity from LTs}}\, ,\\
    \d Y_t &= \underbrace{\eta \, \d B_t }_{\text{activity from LPs}} \underbrace{+ \xi\, \d \Nb_t - \xi\, \d \Na_t}_{\text{activity from LTs}}\, ,\\
    \d Z_t &= \underbrace{\left[-(\vp_{c_t})'(Y_{t\shortminus}+\xi) +  (\vp_{c_t})'(Y_{t\shortminus}) \right]\, \d \Nb_t + \left[-(\vp_{c_t})'(Y_{t\shortminus}-\xi) +  (\vp_{c_t})'(Y_{t\shortminus}) \right]\, \d \Na_t}_{\text{activity from LTs}}\,,\\
    \d \Na_t &= \ch_{Y_{t\shortminus}>\xi}\, \d \hatNa_t \,,\\
    \d \Nb_t &= \d \hatNb_t\,,
\end{align*}
with $\eta>0$ representing the volatility in liquidity provision, $\Nb_0=\Na_0 =0$, $X_0 \in \R^+$, $Y_0\in(\xi,\infty)$, and $Z_0=X_0/Y_0$. 
The constraint $\ch_{\{Y_{t\shortminus} > \xi\}}$ is to guarantee that there is a strong solution to the above dynamics; in  practical applications this indicator is always equal to one. The interpretation of the above equations is as follows, given a liquidity taking buy order (resp.~sell order), the quantity $Y_t$ goes down (resp.~up) by $\xi$, the quantity in $X_t$ increases (resp.~decreases) to its new level $\vp_{c_t}(Y_{t\shortminus} - \xi)$ (resp.~$\vp_{c_t}(Y_{t\shortminus} + \xi)$), and the instantaneous rate moves to its new level $\vp_{c_t}'(Y_{t\shortminus}-\xi)$ (resp.~$\vp_{c_t}'(Y_{t\shortminus}+\xi)$). \\

Finally, we endow the space $(\Omega, \mathcal F )$ with the $\mathbb P-$augmented canonical filtration $\mathbb F := \left(\mathcal F_t  \right)_{t\in \mfT} $ generated by $ (\chi_t)_{t \in \mfT}$.\\

\subsection{The weak formulation}

We introduce $(L_t)_{t\in\mfT}$ as the Doléans-Dade exponential of
$$\int_0^t\, (\lambdaa_u - a_0)\, \d \tildeNa_u + \int_0^t\, (\lambdab_u - a_0)\, \d \tildeNb_u\,, $$
where
\begin{align}
    \lambdaab_t = \barlambdaab(Z_{t\shortminus}, Y_{t\shortminus}, S_t)\,,\qquad 
      \d \tilde{N}^\pm_t = \d \hat{N}^\pm_t - \lambda^\pm_t\, \d t\,,
\end{align}
with
\begin{align}\label{eq:intensities}
    \barlambdaab(Z, Y, S) &= \max\left\{a_0, a_1 + a_2\,Y \pm a_3\,(Z - S)\right\}\,,\qquad a_0>0, \,\,\, a_1, a_2,a_3\geq 0\,.
\end{align}
The linear model for the intensities (inspired by \cite{stoikov2009option}) tries to capture three stylized facts: (i) $a_1$ is the baseline intensity of order arrivals of liquidity motivated trades, (ii) $a_2$ models the relationship between arrivals and depth (locally, all else being the same, higher depth may imply more arrivals),\footnote{The rationale for this is that more depth implies less slippage for any liquidity taking trade executed in the pool. Thus, all else being the same, it is more desirable to trade in a pool with more depth. For mathematical tractability, we employ the quantity $a_2\,Y_t$ as a proxy for the depth. This is a reasonable approximation over a short time horizon, during which large moves in $Y_t$ are mainly driven by LPs activity.} and (iii) $a_3$ captures the impact of trading by the arbitrageurs that align the quotes between the external LOB venue and the CFM.  The effect of $S-Z$ in the intensities is asymmetric, this is because when $S>Z$ (resp.~$S<Z$), all else being equal,  we expect a higher (resp.~lower) buying intensity and a lower (resp.~higher) selling intensity given that the instantaneous rate in the pool is underpriced (resp.~overpriced). The minimum intensity $a_0>0$ is a technical condition to keep the intensity away from zero. This is non-restrictive and in practice we expect both intensities to be bounded away from zero. In Subsection \ref{sec:approx} we carry out experiments with market data to support these claims.\\

Let $\nu_\infty >0$ be a bound for the speed at which the representative LP adds or subtracts liquidity from the pool. 
The set of admissible strategies of the representative LP is
\begin{equation}
    \mcA := \left\{ \nu=(\nu_t)_{t\in\mfT}:\,  \nu \text{ is $\Fb$-progressively measurable and } |\nu_t| \le \nu_\infty \; \mathbb P \otimes \d t \text{ a.e.} \right\}\,,
\end{equation}
and for each $\nu \in \mcA$ we introduce the Radon-Nikodym derivative $(K^\nu_t)_t$ given by
\[
K^\nu_t := \exp \Bigg(  \int_0^t \frac{\nu_s}{\eta} \mathrm dB_s - \frac 12 \int_0^t  \left(\frac{\nu_s}{\eta}\right)^2 \mathrm ds \Bigg).
\]
Given that $\nu\in\mcA$, the process $(K^\nu_t)_{t\in\mfT}$ is a martingale (by Novikov's condition using that $\nu$ is bounded) and we define the probability measure $\mathbb P^\nu$ given by
\[
\frac{\mathrm d\mathbb P^\nu}{\mathrm d\mathbb P} := K^\nu_T\, L_T\,.
\]
Given that all the probability measures $(\mathbb P^\nu)_{\nu\in\mcA}$ are equivalent, we use the notation \textit{a.s.}~for almost surely without ambiguity.\\

In summary, under $\mathbb P^{\nu},$ the processes $\hatNb$ and $\hatNa$ have respective intensities $\left(\lambdab_t \right)_{t \in \mfT}$ and $\left(\lambdaa_t \right)_{t \in \mfT}$, 
the process
$$B^\nu_t = B_t - \int_0^t \frac{\nu_s}{\eta}\d s$$
is a standard Brownian motion, and $W$ is a standard Brownian motion with $\langle W, B^\nu \rangle = 0$. Thus, under $\mathbb P^{\nu}$ we write
\begin{align*}
    \d S_t &= \sigma\,\d W_t\,,\\
    \d X_t &=  \underbrace{ Z_t\,\nu_t \d t }_{\text{representative LP}} +  \underbrace{Z_t\, \eta \, \d B^\nu_t}_{\text{other LPs}} \quad  \underbrace{ -\left[ \varphi_{c_t}(Y_{t\shortminus}) - \varphi_{c_t}(Y_{t\shortminus} + \xi) \right]\, \d \Nb_t + \left[ \varphi_{c_t}(Y_{t\shortminus} - \xi) - \varphi_{c_t}(Y_{t\shortminus}) \right]\, \d \Na_t}_{\text{activity from LTs}}\, ,\\
    \d Y_t &= \underbrace{\nu_t \d t}_{\text{representative LP}}  + \underbrace{\eta \, \d B^\nu_t}_{\text{other LPs}} \quad  \underbrace{+ \xi\, \d \Nb_t - \xi\, \d \Na_t}_{\text{activity from LTs}}\, ,\\
    \d Z_t &= \underbrace{\left[-(\vp_{c_t})'(Y_{t\shortminus}+\xi) +  (\vp_{c_t})'(Y_{t\shortminus}) \right]\, \d \Nb_t + \left[-(\vp_{c_t})'(Y_{t\shortminus}-\xi) +  (\vp_{c_t})'(Y_{t\shortminus}) \right]\, \d \Na_t}_\text{activity from LTs}\, .
\end{align*}

\subsection{The problem of the follower: the representative LP}\label{sec:representative-LP-problem}

The representative LP finances her stake in the pool by trading in the external market. 
That is, to add (resp.~subtract) $\nu_t\,\d t$ units of $Y$ in (resp.~from) the pool at time $t$, the representative LP purchases (resp.~sells) $\nu_t\,\d t$ units in the external market at price $(S_t + \tempLP\,\nu_t)\,\nu_t\,\d t$ where $\tempLP$ is the temporary price impact in the external venue. 
Thus, for $0\le s \le t \le T$, the performance criterion of the representative LP is defined as
\begin{equation}
    \begin{split}
        \rwealth_{s,t}^\nu :&= \underbrace{ \left(Z_t\, Y_t + Y_t\, S_t\right)  - \left(Z_s\, Y_s + Y_s\, S_s\right)}_{\left(X_t + Y_t\,S_t\right) - \left(X_s + Y_s \, S_s\right)} - \int_s^t Z_{u\shortminus}\, \nu_{u\shortminus}\, \d u - \int_s^t (S_u + \tempLP\, \nu_{u\shortminus})\, \nu_{u\shortminus}\, \d u\;.
    \end{split}
\end{equation} 
The first two terms account for the mark-to-market of the stake in the pool, with the quantity $Y_t$ valued at the external price  $S_t$. The term $- \int_s^t Z_{u\shortminus}\, \nu_{u\shortminus}\, \d u$ is the financing of the units added or subtracted in asset $X$.\footnote{The calculations below can be carried out if one also includes a running penalty of the form $-\int_s^t p(u,Z_u,Y_u,S_u)\d u$. For simplicity of formulae, we do not include the running penalty.} Lastly, the term $- \int_s^t (S_u + \tempLP\, \nu_{u\shortminus})\, \nu_{u\shortminus}\, \d u$ is the financing of the units added or subtracted in asset $Y$ ---the representative  LP pays the costs of trading (e.g., walking the LOB and spread) through the quadratic term $\tempLP\,\nu^2_t$.\\

The above can be written as
\begin{equation}
    \begin{split}
        \rwealth_{s,t}^\nu &=   - \int_s^t \tempLP\, \nu_{u\shortminus}^2 \, \d u + \eta\, \int_s^t (S_{u\shortminus} + Z_{u\shortminus})\, \d B_u^{\nu} \\
        &\qquad\qquad + \sigma\, \int_s^t Y_{u\shortminus}\, \d W_u 
        + \int_s^t \Deltarwealtha_u\, \d \Na_u
        + \int_s^t \Deltarwealthb_u\, \d \Nb_u\,,
    \end{split}
\end{equation}
for $s<t$  where
\begin{equation}
    \Deltarwealthab_u := \pm\, \xi\, \left( S_{u} - \frac{Z_{u\shortminus}\, Y_{u\shortminus}}{Y_{u\shortminus} \pm\, \xi} \right).
\end{equation}
Thus, the control problem that the representative LP wishes to solve is given by
\begin{equation}
    V_t(\rew) = \sup_{\nu \in \mcA} \E_t^\nu \left[ -\exp \left\{ - \gamma\, (\rew + \rwealth^\nu_{t,T} )\right\}  \right]\,,
\end{equation}
where $\rew\in L^2(\Omega, \mcF_T, \Pb^\nu)$ is the reward offered by the venue to the representative LP and $\gamma>0$ is the risk aversion parameter of the representative LP.\footnote{$\E_t^\nu$ represents the conditional expectation with respect to the sigma-algebra $\mcF_t$ under the probability measure $\Pb^\nu$.} To ensure that the above expectation is well-defined, we need the following technical condition.
\begin{cond}\label{assumptionRgamma}
There exists $\gamma'>\gamma$ such that 
   \begin{align}
    \underset{\nu \in \mathcal A}{\sup}\; \mathbb E^\nu \left[ \exp \{ - \gamma' \rew\} \right] < +\infty \,.
\end{align} 
\end{cond}

In Section \ref{sec:solution-LP-problem}, for each reward $\rew$, we  prove that there exists a unique optimal response $\nu^*(\rew)$.

\subsection{The problem of the leader: the venue}\label{sec:venue-problem}

The venue provides a contract $\rew$ to the representative LP and wishes to attract as much order flow as possible to the venue (i.e., $\Na_T + \Nb_T$). Thus,  the venue wishes to maximize the following performance criterion
\begin{equation}
    \E^{\nu^*(\rew)} \left[-\exp\left\{-\zeta\, \left(\feeV\,(\Na_T + \Nb_T) - \rew\right)\right\} \right]\,,
\end{equation}
where $\zeta>0$ is the venue's risk aversion parameter and $\feeV>0$ is a constant transaction fee paid by LTs to the venue. For the above expectation to be well-defined, we need the following technical condition.
\begin{cond}\label{assumptionReta}
There exists $\zeta'>\zeta$ such that 
   \begin{align}
    \underset{\nu \in \mathcal A}{\sup}\; \mathbb E^\nu \left[ \exp \{ \zeta'\, \rew\} \right] < +\infty\,. 
\end{align} 
\end{cond}

The set of admissible contracts is given by
\begin{equation}
    \mcA^R = \left\{ \rew,\; \mathcal F_T-\text{measurable, such that } V_{0}(\rew) \ge R \text{ and Conditions } \ref{assumptionRgamma} \text{ and } \ref{assumptionReta} \text{ are satisfied}\right\}\,,
\end{equation}
where $R<0$ is the reservation level of the LP, i.e. the smallest acceptable utility level, in the sense that the LP refuses the contract if it doesn't allow them to reach this utility level; see Remark \ref{rem:R} below for a more detailed discussion.\\

Next, we characterize the solution to the leader-follower game between the venue and the representative LP. Section \ref{sec:solution-LP-problem} studies the optimal response of the follower $\nu^*(\rew)$ and Section \ref{sec:solution-venue-problem} finds the optimal contract $\rew^*$.

\section{Solving the problem of the LP}\label{sec:solution-LP-problem}

For any $(\nu, A) \in [-\nu_\infty, \nu_\infty] \times \mathbb R^4 $, with $A = \left(\AW, \AB, \Ab, \Aa \right)$, we define
 $$h(\nu, A) =  -\tempLP\, \nu^2 + \frac 1{\eta} \AB \nu\,.$$

It is easy to see that for any $A\in \R^4$, the maximizer of $h(\cdot, A)$ is reached at 
$$\bar{\nu}(A) = \left( \frac{\AB}{2\,\tempLP\, \eta} \vee -\nu_{\infty}\right)\wedge \nu_\infty\,.$$

We then define
$$H(A) = \underset{|\nu| \le \nu_{\infty} }{\sup} h(\nu, A) = h(\bar{\nu}(A), A)\,.$$

For any constant $P_0 \in \mathbb R$ and predictable process $(A_t)_{t\in\mfT} = \left(\AW_t, \AB_t, \Ab_t, \Aa_t \right)_{t\in\mfT}$ satisfying  that
$$\sup_{\nu \in \mcA} \E^{\nu} \left[\int_0^T \left(\AB_t\right)^2 + \left(\AW_t\right)^2\, \d t\right] < \infty\,,$$
\noindent and 
\begin{align}
     \int_0^T  &\Bigg|   \frac{\gamma}{2}\, \left[ \left(\AW_s + \sigma\, Y_{s\shortminus}\right)^2 + \left(\AB_s + \eta\, (S_{s\shortminus} + Z_{s\shortminus})\right)^2 \right] \\
     &\qquad\quad - \sum_{i\in\{+,-\}} \frac{\lambdai_s\,(1 - e^{-\gamma(\Ai_s + \Deltarwealthi_s)})}{\gamma}  - H(A_s)\Bigg|\, \d s < \infty \qquad \Pb\text{-a.s.}\,,
\end{align}
we introduce the process
\begin{align}
    P_t^{P_0, A} &= P_0 + \sum_{i\in\{+,-\}} \int_0^t  \Ai_s\, \d \hat \Ni_s + \int_0^t \AW_s\, \d W_s +\int_0^t \AB_s\, \d B_s \\
    &+ \int_0^t  \Bigg\{   \frac{\gamma}{2}\, \left[ \left(\AW_s + \sigma\, Y_{s\shortminus}\right)^2 + \left(\AB_s + \eta\, (S_{s\shortminus} + Z_{s\shortminus})\right)^2 \right] - \sum_{i\in\{+,-\}} \frac{\lambdai_s\,(1 - e^{-\gamma(\Ai_s + \Deltarwealthi_s)})}{\gamma} - H(A_s)\Bigg\}\, \d s\,,
\end{align}
and we denote by $\Lambda$ the set of all processes $(A_t)_{t\in\mfT}$ such that Conditions \ref{assumptionRgamma} and \ref{assumptionReta} are satisfied with reward $\rew=P_T^{0, A}$ and

\begin{equation}
    \sup_{\nu \in \mcA} \sup_{t\in \mfT} \E^\nu \left[ \exp\left\{ -\gamma'\, P_t^{0,A} \right\}  \right] < \infty \quad \text{ for some } \gamma'>\gamma\,.
\end{equation}

We denote by $\mathcal P$ the set
$$\mathcal P = \left\{P_T^{P_0, A} \text{ such that } P_0\in \mathbb R,\; A\in \Lambda,\; \text{and } V_{0}\left( P_T^{P_0, A} \right)\ge R\right\}.$$

It is clear that $\mathcal P \subset \mathcal A^R.$ The next two theorems are proved in \ref{proofcontractrep}.

\begin{theorem}\label{contractrep}
    For any $\rew \in \mathcal A^R$, there exists a unique $(P_0, A)\in \mathbb R \times \Lambda$ such that $\rew = P_T^{P_0, A}.$ In particular, $\mathcal A^R = \mathcal P.$
\end{theorem}

\begin{theorem}\label{theorem: contract rep value}
    
    For any reward structure $\rew = P_T^{P_0, A}$ as stated in Theorem \ref{contractrep}, the LP's value function is
    \begin{equation}
        V_0\big(P_T^{P_0,A}\big) = -\exp(-\gamma\, P_0)\,,
    \end{equation}
    with optimal liquidity providing speed
    \begin{equation}
        \nu^*_t = \bar{\nu}(A_t) = \left( \frac{\AB_t}{2\,\tempLP\, \eta} \vee -\nu_{\infty}\right)\wedge \nu_\infty\,.
    \end{equation}
\end{theorem}

\begin{rem}\label{rem:R}
    Theorem \ref{theorem: contract rep value} shows that $P_0$ is the certainty equivalent of an  LP that faces contract $P^{P_0,A}_T$ and trades optimally; it also shows that any admissible contract must satisfy $-\exp(-\gamma P_0) \ge R$. In papers such as \cite{euch2021optimal}, the reservation level $R$ is chosen to be the utility level of the agent in the absence of a contract. This is reasonable  in their framework since an agent behaving in an optimal way is already profitable even without a contract. This is not the case in our paper: it is well-known that, in the absence of fees, an agent providing liquidity for a CPM is exposed to a concave and negative payoff. Therefore, if the venue wants to attract LPs, it should redistribute some of the fees. More precisely, since the LP can simply choose not to participate in the pool, if the venue wants to retain them it should at least offer a contract that will allow the LP to have a non-negative certainty equivalent, i.e., $P_0$ should be non-negative. It is therefore reasonable to consider $-1 \le R <0$. Competition between venues in order to attract LPs should tend to increase this value.  In the numerical examples below, we choose $R$ so that the representative LP and the venue share the profits collected from fees.
\end{rem}

\section{Solving the problem of the venue}\label{sec:solution-venue-problem}
Here, we formulate the control problem of the venue. First, we discuss the risk-neutral case in Subsection \ref{subsec: risk neutral} and then the exponential utility case in Subsection \ref{subsec: exponential utility}.

\subsection{Risk-neutral venue}\label{subsec: risk neutral}

We first consider the case where $\zeta\downarrow 0$, in which case the venue's optimization problem amounts to
\begin{equation}
    \sup_{\rew \in \mcA^R } \E^{\nu^*(\rew)} \left[\feeV\,(\Na_T + \Nb_T) - \rew \right]\,.
\end{equation}

In view of Theorem \ref{theorem: contract rep value}, this is equivalent to 
\begin{equation}
    \sup_{P_0\ge \hat P_0}\sup_{A \in \Lambda}  \E^{\nu^*(P_T^{ P_0, A})} \left[\feeV\,(\Na_T + \Nb_T) - P_T^{\hat P_0, A} \right]\,,
\end{equation}
where $\hat P_0 = -1/\gamma \log(-R)$. For any $A \in \Lambda$, the supremum over $P_0$ of the objective function is reached at $P_0 = \hat P_0$. The problem of the venue is therefore equivalent to
\begin{align}
    \sup_{A \in \Lambda} \E^{\nu^*(P_T^{\hat P_0, A})} \Bigg[&\int_0^T  \Bigg\{ - \frac{\gamma}{2}\, \left[ \left(\AW_s + \sigma\, Y_{s\shortminus}\right)^2 + \left(\AB_s + \eta\, (S_{s\shortminus} + Z_{s\shortminus})\right)^2 \right] + \sum_{i\in\{+,-\}} \frac{\lambdai_s\,(1 - e^{-\gamma(\Ai_s + \Deltarwealthi_s)})}{\gamma} \\
    &\qquad \qquad  + H(A_s) - \frac{\AB_s\, \nu^*_s}{\eta} \Bigg\}\, \d s +\sum_{i\in\{+,-\}} \int_0^T  (\feeV - \Ai_s)\, \d \Ni_s - \int_0^T  \Aa_s\,\ch_{Y_{s\shortminus}\leq \xi}\, \d \hat{N}^-_s\Bigg]\,.
\end{align}
 We define
\begin{align}
    v(t, z, y, s) = \sup_{A \in \Lambda} &\E_{t,z,y,s}^{\nu^*(P_T^{\hat P_0, A})} \Bigg[\int_t^T  \Bigg\{   - \frac{\gamma}{2}\, \left[ \left(\AW_s + \sigma\, Y_{s\shortminus}\right)^2 + \left(\AB_s + \eta\, (S_{s\shortminus} + Z_{s\shortminus})\right)^2 \right]\\
    & \hspace{3cm} + \sum_{i\in\{+,-\}} \frac{\lambdai_s\,(1 - e^{-\gamma(\Ai_s + \Deltarwealthi_s)})}{\gamma} + H(A_s)- \frac{\AB_s\, \nu^*_s}{\eta}\Bigg\}\, \d s \\
    & \hspace{3cm}\qquad  +\sum_{i\in\{+,-\}} \int_t^T  (\feeV - \Ai_s)\, \d \Ni_s - \int_0^T  \Aa_s\,\ch_{Y_{s\shortminus}\leq \xi}\, \d \hat{N}^-_s \Bigg]\,.
\end{align}

We derive the Hamilton–Jacobi–Bellman (HJB) equation and substitute $H(A) = h(\bar{\nu}(A), A)$ to obtain
\begin{align}
    0 &= \barlambdaa\, \left[ \feeV + v(t, \Za, Y - \xi, S) - v(t, Z, Y, S) \right]\, \ch_{Y > \xi} + \barlambdab\, \left[ \feeV + v(t, \Zb, Y + \xi, S) - v(t, Z, Y, S) \right]\\
    &\quad  + \frac{1}{4} \Bigg[ -2\, \gamma\,  \eta^2\, S^2 - 4\, \gamma\,  \eta^2\, S\, Z - 2\, \gamma\,  \eta^2\, Z^2 - 2\, \gamma\,  \sigma^2\, Y^2 + 2\, \sigma^2\, \partial_{SS}v + 2\, \eta^2\, \partial_{YY}v + 4\, \partial_t v\Bigg] \\
    &\quad + \sup_{A} \Bigg\{ \bigg[-\frac{\gamma}{2}\, (\AW)^2 - Y\, \gamma\,  \sigma\, \AW\bigg] \\
    &\hspace{2cm} + \Bigg[- \frac{\gamma}{2}\, (\AB)^2 + h(\bar{\nu}(A),A) - \frac{\AB\, \bar{\nu}(A)}{\eta} -(S+Z)\,\gamma\,\eta\, \AB + \partial_Y v\, \bar{\nu}(A)\Bigg]\\
    &\hspace{2cm}  + \barlambdaa(Z, Y, S)\,  \left[-\Aa + \frac{1-e^{-\gamma\, \left(\Aa + \Deltarwealtha\right)}}{\gamma}\right]\, \ch_{Y > \xi} +  \barlambdaa\,  \left[-\Aa + \frac{1-e^{-\gamma\, \Aa}}{\gamma}\right]\, \ch_{Y \leq \xi}\\
    &\hspace{2cm}  + \barlambdab(Z, Y, S)\,  \left[-\Ab + \frac{1-e^{-\gamma\, \left(\Ab + \Deltarwealthb\right)}}{\gamma}\right]\Bigg\}\,, 
    \label{eq:hjb risk neutral}
\end{align}

where $\Za = Z\, Y^2/(Y-\xi)^2$, $\Zb = Z\, Y^2/(Y+\xi)^2$. We define a new function
$$\alpha(t, Z, Y, S) = \frac{\frac{\partial_Y v}{\tempLP\, \eta} - 2\,(S+Z)\,\gamma\,\eta}{2\, \gamma + \frac{1}{\tempLP\, \eta^2}}\,.$$

Then the supremum in the HJB equation is attained at
\begin{align}
    (\AW)^* &= - Y\, \sigma\,, \label{eq:AW opt}\\
    (\AB)^* &= \begin{cases}
        \alpha(t, Z, Y, Z) \quad & \text{if } \alpha(t, Z, Y, S)\in [-2\, \tempLP\, \eta\, \nu_\infty, 2\, \tempLP\, \eta\, \nu_\infty]\,,\\
        -(S+Z)\, \eta \quad & \text{otherwise}\,,
    \end{cases}\, \label{eq:AB opt}\\
    (\Aa)^* &= - \Deltarwealtha\,\ch_{Y > \xi}\,, \label{eq:Aa opt}\\
    (\Ab)^* &= - \Deltarwealthb\,. \label{eq:Ab opt}
\end{align}
Observe that with this optimal reward structure, the venue compensate any risk that the LPs face from the external price movement and liquidity taking.
The PDE obtained from the HJB equation \eqref{eq:hjb risk neutral} and substituting $(\Aa, \Ab, \AW, \AB)$ as in \noeqref{eq:AB opt} \noeqref{eq:Aa opt} \eqref{eq:AW opt}--\eqref{eq:Ab opt} is difficult to solve, and we propose in Section \ref{sec:approx} an approximation technique. However, following \cite{phamcontrol}, we employ the below verification theorem assuming that the solution exists in the classical sense. 

\begin{proposition}
    If there exists a function $v\in C^{1,2}([0,T]\times \R^3; \R)$ such that $v$ has quadratic growth in $(Z,Y,S)$ (uniformly in $t$) and it satisfies the HJB equation \eqref{eq:hjb risk neutral} with $v(T, \cdot, \cdot, \cdot) = 0$, then
    $$v(t, Z_t, Y_t, S_t) = \sup_{A \in \Lambda } \E_t^{\nu^*(P_T^{P_0,A})} \left[\feeV\,(\Na_T + \Nb_T) - P_T^{P_0,A} \right] \quad \d t \otimes \d \Pb\,\text{-a.e.}\,,$$
     with optimal $(\Aa)^*, (\Ab)^*, (\AW)^*, (\AB)^*$ as in \eqref{eq:AW opt}--\eqref{eq:Ab opt}.
\end{proposition}

\subsection{Risk-averse venue}\label{subsec: exponential utility}

Observe that under $\Pb^\nu$, we have

\begin{align}
    P_t = P_0 &+ \int_0^t  \Bigg\{ \frac{\gamma}{2}\, \left[ \left(\AW_s + \sigma\, Y_{s\shortminus}\right)^2 + \left(\AB_s + \eta\, (S_{s\shortminus} + Z_{s\shortminus})\right)^2 \right] \\
    &\qquad \qquad - \sum_{i\in\{+,-\}} \frac{\lambdai_s\,(1 - e^{-\gamma(\Ai_s + \Deltarwealthi_s)})}{\gamma} - H(A_s) + \frac{\AB_s\, \bar \nu(A_s)}{\eta}\Bigg\}\, \d s \\
    &+\sum_{i\in\{+,-\}} \int_0^t  \Ai_s\, \d \Ni_s + \int_0^t \Aa_s\, \ch_{Y_{s\shortminus} \leq \xi}\, \d \hat{N}^-_s + \int_0^t \AW_s\, \d W_s +\int_0^t \AB_s\, \d B_s^{\nu^*} \,.
\end{align}

We let $V$ be the value function of the risk averse venue, given by
\begin{equation}
    \Vvn(t, Z_t, Y_t, S_t, \hatNa_t, \Na_t, \Nb_t, P_t^{P_0,A}) = \sup_{A\in \Lambda} \E^{\nu^*(\rew)}_{t} \left[-\exp\left\{-\zeta\, \left(\feeV\,(\Na_T + \Nb_T) - \rew \right)\right\} \right]\,.
\end{equation}

The full HJB equation for the above control problem is
\begin{equation}
    \begin{split}
        0 &=\sup_{\Aa, \Ab, \AW, \AB} \Bigg( \partial_t \Vvn + (\partial_Y \Vvn)\, \bar \nu(A) + (\partial_P \Vvn)\, \Bigg\{ \frac{\gamma}{2}\, \left[ \left(\AW + \sigma\, Y\right)^2 + \left(\AB + \eta\, (S + Z)\right)^2 \right]\\
        &\qquad  - \sum_{i\in\{+,-\}} \frac{\barlambdai(Z, Y, S)\,(1 - e^{-\gamma(\Ai + \Deltarwealthi)})}{\gamma} - H(A) + \frac{\AB\, \bar \nu(A)}{\eta}\Bigg\} + (\partial_{YY} \Vvn)\, \frac{\eta^2}{2} + (\partial_{YP} \Vvn)\, \AB\, \eta\\
        &\qquad + (\partial_{SS} \Vvn)\, \frac{\sigma^2}{2} + (\partial_{SP} \Vvn)\, \AW\, \sigma + (\partial_{PP} \Vvn)\, \frac{(\AW)^2+(\AB)^2}{2} + \sum_{i\in\{+,-\}} \barlambdai(Z, Y, S)\, \left[ \Vvni - \Vvn \right]\Bigg)\,,
    \end{split}
\end{equation}
with
\begin{align}
    \Vvna &= \Vvn(t, Z\, Y^2/(Y-\xi\ch_{Y > \xi})^2, Y - \xi\ch_{Y > \xi}, S, \hatNa+1, \Na + \ch_{Y > \xi}, \Nb, P+\Aa)\,,\\
    \Vvnb &= \Vvn(t, Z\, Y^2/(Y+\xi)^2, Y+\xi, S, \hatNa, \Na, \Nb + 1, P+\Ab)\,.
\end{align}

By taking the ansatz
\begin{equation}
    \Vvn(t, Z, Y, S, \hatNa,\Na, \Nb, P) = -e^{-\zeta[\feeV(\Na+\Nb)-P+v(t,Z,Y,S)]}\,
\end{equation}
and dividing the HJB equation by $\zeta\, \Vvn(t, Z, Y, S, \hat{\Na}, \Na, \Nb, P)$, we obtain
\begin{align}
    0 &= \quad \frac{1}{2} \inf_{\AB} \Bigg\{ (\gamma+\zeta)\, (\AB)^2 + 2\,\tempLP\, \bar \nu(A)^2 + 2\,\left[\gamma\, \eta\, (S+Z) - \zeta\,\gamma\,\partial_Y v \right]\, \AB - 2\,(\partial_Y v)\, \bar \nu(A)\Bigg\}\\
    &\quad + \frac{1}{4} \inf_{\AW} \Bigg\{ 2\,(\gamma+\zeta)\,(\AW)^2 + 4\,\sigma\,(\gamma\, Y - \zeta\, \partial_S v)\, \AW \Bigg\}\\
    &\quad - \sup_{\Aa} \left\{ \barlambdaa(Z, Y, S)\, \left( \frac{1 - e^{-\zeta\,\left(\va - v + \feeV - \Aa \right)}}{\zeta}\, \ch_{Y>\xi} + \frac{1 - e^{\zeta\, \Aa}}{\zeta}\, \ch_{Y\leq\xi}  + \frac{1-e^{-\gamma(\Aa+\Deltarwealtha)}}{\gamma} \right)\right\}\\
    &\quad - \sup_{\Ab} \left\{ \barlambdab(Z, Y, S)\, \left( \frac{1 - e^{-\zeta\,\left(\vb - v + \feeV - \Ab \right)}}{\zeta} + \frac{1-e^{-\gamma(\Ab+\Deltarwealthb)}}{\gamma} \right)\right\}\\
    &\quad + \frac{1}{2}\, \gamma\, \eta^2\, S^2 + \gamma\, \eta^2\, S\, Z + \frac{1}{2}\, \gamma\, \eta^2\, Z^2 + \frac{1}{2}\, \gamma\, \sigma^2\, Y^2\\
    &\quad + \frac{1}{2}\, \zeta\, \sigma^2\, (\partial_S v)^2 - \frac{1}{2}\, \sigma^2\, \partial_{SS} v + \frac{1}{2}\, \zeta\, \eta^2\, (\partial_Y v)^2 - \frac{1}{2}\,\eta^2\, \partial_{YY} v - \partial_t v\,,
    \label{eq:hjb risk averse}
\end{align}
with
\begin{align}
    \va &= v(t, Z\,Y^2/(Y-\xi\ch_{Y>\xi})^2,Y-\xi\ch_{Y>\xi}, S)\,,\\
    \vb &= v(t, Z\,Y^2/(Y+\xi)^2,Y+\xi, S)\,
\end{align}
and terminal condition $v(T,Z,Y,S)=0$.\\

We define 
$$\alpha_\zeta(t, Z, Y, S) = \frac{\left( \frac{1}{\tempLP\, \eta} + 2\, \zeta\, \eta \right)\,\partial_Y v - 2\, \gamma\, \eta\, (S+Z)}{2\,(\gamma+\zeta) + \frac{1}{\tempLP\, \eta^2}}\,.$$

The optimal values for $\AB$, $\AW$, $\Aa$, and $\Ab$ are
\begin{align}
    (\AB_t)^* &= \begin{cases} \alpha_\zeta(t, Z_{t\shortminus}, Y_{t\shortminus}, S_t) & \text{if } \alpha_\zeta(t, Z_{t\shortminus}, Y_{t\shortminus}, S_t)\in [-2\,\tempLP\,\eta\,\nu_\infty, 2\,\tempLP\,\eta\,\nu_\infty]\\
    \frac{\zeta\,\gamma\, \partial_Y v - \gamma\, \eta\, (S+Z)}{\gamma+\zeta} & \text{otherwise}
    \end{cases}\,, \label{eq:AB opt ra}\\
    (\AW_t)^* &= \frac{\sigma\,( \zeta\, \partial_S v-\gamma\, Y)}{\gamma+\zeta}\,, \label{eq:AW opt ra}\\
    (\Aa_t)^* &= \begin{cases}
        \frac{-\gamma\, \Deltarwealtha+\zeta\,(\va-v+\feeV)}{\gamma+\zeta} & \text{if } Y_{t\shortminus}>\xi,\\
        0 & \text{otherwise,}
    \end{cases}\, \label{eq:Aa opt ra}\\
    (\Ab_t)^* &= \frac{-\gamma\, \Deltarwealthb+\zeta\,(\vb-v+\feeV)}{\gamma+\zeta}\,. \label{eq:Ab opt ra}
\end{align}

As before, we state the following standard verification theorem.
\newpage
\begin{proposition}
    If there exists a function $v\in C^{1,2}([0,T]\times \R^3; \R)$ such that $v$ has quadratic growth in $(Z,Y,S)$ (uniformly in $t$) and it satisfies the HJB equation \eqref{eq:hjb risk averse} with $v(T, \cdot, \cdot, \cdot) = 0$, then
    $$v(t, Z_t, Y_t, S_t) = \sup_{A \in \Lambda } \E_t^{\nu^*(P_T^{P_0,A})} \left[-\exp \left\{ - \zeta \left(\feeV\,(\Na_T + \Nb_T) - P_T^{P_0,A} \right) \right\} \right] \quad \d t \otimes \d \Pb\,\text{-a.e.}\,,$$
     with optimal $(\Aa)^*, (\Ab)^*, (\AW)^*, (\AB)^*$ as in \eqref{eq:AB opt ra}--\eqref{eq:Ab opt ra}. \noeqref{eq:AW opt ra} \noeqref{eq:Aa opt ra}
\end{proposition}

\subsection{Approximate closed-form solutions}\label{sec:approx}

The HJB equations above are difficult to solve. Next, we carry out approximations that make the above equations more tractable. 

\subsubsection{Risk-neutral venue}

First, we ignore the boundary imposed on the control, that is, we take $\nu_\infty\uparrow \infty$ so that $(\AB_t)^*=\alpha(t, Z_t, Y_t, S_t)$. Furthermore, we only consider the case where $Y > \xi$. Substituting this, we obtain
\begin{align}
    0 &= \sum_{i\in\{+,-\}}  \lambdai\, \left[ \feeV + h(t, \Zi, Y + \deltai\, \xi, S) - h(t, Z, Y, S) \right] \\
    &\phantom{{}={}} + \frac{1}{4} \Bigg[  -2\, \gamma\,  \eta^2\, S^2 - 4\, \gamma\,  \eta^2\, S\, Z - 2\, \gamma\,  \eta^2\, Z^2
      + 2\, \sigma^2\, \partial_{SS}h + 2\, \eta^2\, \partial_{YY}h + 4\, \partial_t h\Bigg] \\
    &\phantom{{}={}} + \frac{1}{4}\, \Bigg\{\frac{\left(\frac{\partial_Y h}{\tempLP\, \eta} - 2\,(S+Z)\,\gamma\,\eta\right)^2}{2\, \gamma + \frac{1}{\tempLP\, \eta^2}} + 4\,  \sum_{i\in\{+,-\}} \barlambdai(Z, Y, S)\, \bigg[  \deltai\, \xi\, \left( S - \frac{Z\, Y}{Y + \deltai\, \xi} \right) \bigg]\Bigg\}\,, \label{eq: hjb venue}
\end{align}
with $\deltaa = -1$, and $\deltab = 1$.
Next, we remove the indicator $\ch_{Y_t>\xi}$ from the   stochastic intensities to obtain 
\begin{align}
    \lambdaa_t & \approx \max\left\{a_0, a_1 + a_2\,Y_t + a_3\,(S_t - Z_t)\right\}\,,\qquad    \lambdab_t \approx \max\left\{a_0, a_1 + a_2\,Y_t + a_3\,(Z_t - S_t)\right\}\,.
\end{align}
We have four possibilities: (i) $\lambdaa, \lambdab>a_0$, (ii) $\lambdaa = a_0,\, \lambdab>a_0$, (iii)  $\lambdaa > a_0,\, \lambdab=a_0$, and (iv) $\lambdaa, \lambdab = a_0$. In what follows, we work under the assumption that model parameters are such that one is always in scenario (i), which helps us to obtain approximate closed-form solutions. First, we justify (with market data) that the above assumption is a good proxy for reality.\\

We employ market data from Uniswap V2 and from Binance for the pair ETH-USDC between 1 January 2022 and 30 April 2022. During these four months, we monitor the exchange rate in Binance (the midprice $S_t$), the instantaneous exchange rate in the Uniswap pool ($Z_t$), and the trading activity (liquidity taking orders) in the pool. We split the data into 10-minute buckets, and within a given 10-minute window, we compute the aggregate of the buy-initiated trades and sell-initiated trades.\footnote{Using a five-minute or a fifteen-minutes window yields similar results.} We think of these aggregates as a rough estimate for the intensities $\lambda^{\pm}_t$. Finally, we perform a linear regression that explains $\lambda^{\pm}_t$ in terms of the average differences $\pm(S_t-Z_t)$; for now, we set $a_2= 0$ for simplicity.\footnote{The coefficients $\hat{a}_1$, $\hat{a}_3$  are both significative with a $p$-value $p\ll 0.01$. }  We use the estimates $\hat{a}_1, \hat{a}_3$ to compute the boundaries $\pm d$ such that $0=\hat{a}_1 + \hat{a}_3\,d$. If $S_t - Z_t < d$ or $S_t - Z_t > - d$ then $\lambdaa_t = 0$ or $\lambdab_t = 0$, respectively.
Figure \ref{fig:violations-intensity} shows the boundaries $\pm d$ (in red dotted vertical lines) together with market data on $S_t - Z_t$.

\begin{figure}[H]
    \centering
    \includegraphics[width=0.5\linewidth]{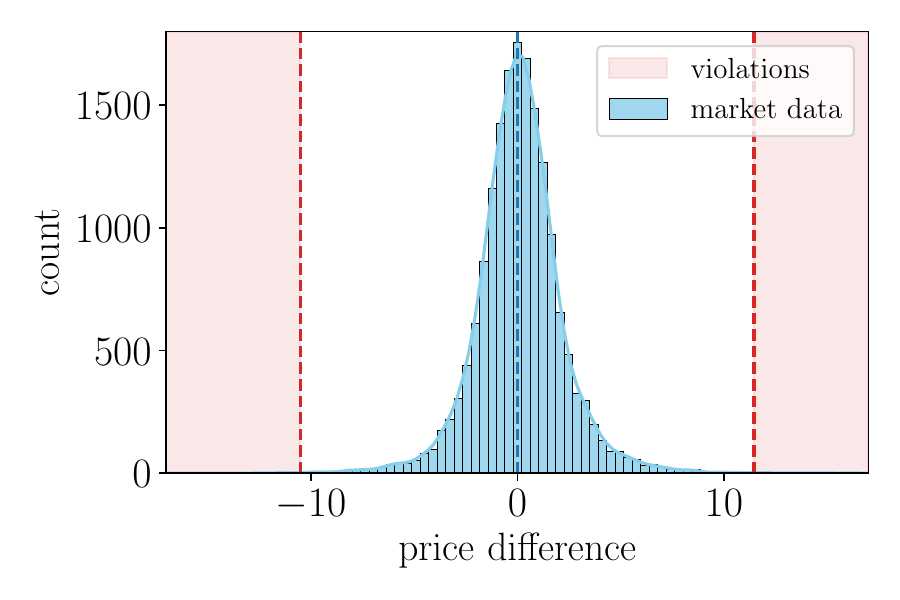}
    \caption{Histogram of price differences between Binance and Uniswap V2 for ETH-USDC between 1 January 2022 and 30 April 2022. The red shaded area represents the region in which $a_1\pm\,a_3 (S_t-Z_t)$ becomes negative. We take $a_2=0$ for simplicity so that the violation boundary is fixed and does not depend on the number of ETH units in the pool.  }
    \label{fig:violations-intensity}
\end{figure}

We find that of the 17,131 data points in the histogram above, the boundaries are breached 30 times to the left and 34 times to the right. Thus,  0.37\% of the data points violate the assumption that in our model $\hat{a}_1 \pm \hat{a}_3\,(S_t - Z_t) > 0$.\\

Next, we use Laurent series to approximate the following terms
    \begin{align}
        \frac{Y}{Y+\deltaab\,\xi} &= \sum_{n=0}^\infty (-1)^n\, \left(\frac{\deltaab\,\xi}{Y}\right)^n = 1 + O\left(\frac{\xi}{Y}\right)\,,\\
        \frac{Y^2}{Y+\deltaab\,\xi} &= Y + \sum_{n=1}^\infty (-1)^n\, \left(\frac{(\deltaab\,\xi)^n}{Y^{n-1}}\right) = Y - \deltaab\,\xi + O\left(\frac{\xi^2}{Y}\right)\,,\\
        \frac{Y^2}{(Y+\deltaab\,\xi)^2} &= \sum_{n=0}^\infty (n+1)\, (-1)^n\, \left(\frac{\deltaab\,\xi}{Y}\right)^n = 1 + O\left(\frac{\xi}{Y}\right)\,,\\
        \frac{Y^3}{(Y+\deltaab\,\xi)^2} &= \sum_{n=0}^\infty Y\, (n+1)\, (-1)^n\, \left(\frac{\deltaab\,\xi}{Y}\right)^n = Y - 2\, \deltaab\, \xi + O\left(\frac{\xi^2}{Y}\right)\,.\\
    \end{align}
Given that $\Zab = Z\, Y^2/(Y+\deltaab\, \xi)^2$, and using the above approximations,  we have that 
    \begin{equation}
        \begin{split}
            h(t, \Zi, Y+\deltaab\,\xi, S) &= h(t, Z, Y+\deltaab\,\xi, S) + Z\,\left( \frac{Y^2}{(Y+\deltaab\,\xi)^2} - 1 \right)\, \partial_Z h(t,Z,Y+\deltaab\,\xi,S)+ O\left(\frac{\xi^2}{Y^2}\right)\\
            &\approx h(t, Z, Y+\deltaab\,\xi, S)\,,\\
            Y\, h(t, \Zi, Y+\deltaab\,\xi, S) &= Y\, h(t, Z, Y+\deltaab\,\xi, S) + Z\,\left( \frac{Y^3}{(Y+\deltaab\,\xi)^2} - Y \right)\, \partial_Z h(t,Z,Y+\deltaab\,\xi,S) + O\left(\frac{\xi^2}{Y}\right)\\
            &\approx Y \,h(t, Z, Y+\deltaab\,\xi, S) - 2\,\deltaab\, \xi\, Z\, \partial_Z h(t, Z, Y+\deltaab\,\xi, S)\,.
        \end{split}
    \end{equation}
Thus, the approximated HJB equation becomes
    \begin{align}
        0 &\approx \partial_t h + \sum_{i\in\{+,-\}} \lambdai\, \left[ \feeV + h(t, Z, Y + \deltai\, \xi, S) - h(t, Z, Y, S) \right]\\
        &\phantom{{}={}} + \frac{1}{4} \Bigg[  -2\, \gamma\,  \eta^2\, S^2 - 4\, \gamma\,  \eta^2\, S\, Z - 2\, \gamma\,  \eta^2\, Z^2 + 2\, \sigma^2\, \partial_{SS}h + 2\, \eta^2\, \partial_{YY}h\Bigg] \\
        &\phantom{{}={}} + \frac{1}{4}\, \Bigg\{\frac{\left(\frac{\partial_Y h}{\tempLP\, \eta} - 2\,(S+Z)\,\gamma\,\eta\right)^2}{2\, \gamma + \frac{1}{\tempLP\, \eta^2}} + 4\,  \sum_{i\in\{+,-\}} \barlambdai(Z, Y, S)\, \bigg[  \deltai\, \xi\, \left(S - Z \right) \bigg]\Bigg\}\\
        &\phantom{{}={}} + a_2\, \sum_{i\in\{+,-\}} \bigg\{ \xi^2\, Z - 2\, \deltai \xi\,Z\, \partial_Z g(t, Z, Y+\deltai\, \xi, S)\bigg\}\,. \label{eq: hjb venue approx 1}
    \end{align}
Let $\hat{h}$ be the approximation of $h$ such that  sign `$\approx$' in \eqref{eq: hjb venue approx 1} turns into an equality. By taking the ansatz $$\hat{h}(t,Z,Y,S)=h_0(t,Y) + h_1(t,Y)\, Z + h_2(t,Y)\,S+ h_3(t)\, Z^2 + h_4(t)\, Z\, S + h_5(t)\, S^2\,,$$ where
    \begin{align}
        h_i(t, Y) = h_{i0}(t) + h_{i1}(t)\, Y, \quad i=0,1,2\,,
    \end{align}
    we obtain
    \begin{align}
        \partial_Y \hat{h} (t, Z, Y, S) &= h_{01}(t) + h_{11}(t)\, Z + h_{21}(t)\, S =  2\, a_2\, \feeV\, (T-t)\,.
    \end{align}
    Substituting this into $(\AB)^*$, we obtain
    \begin{equation}
        (\AB_t)^* \approx \frac{1}{2\, \gamma+ \frac{1}{\tempLP\, \eta^2}}\left({\frac{2\,a_2\,\feeV\, (T-t)}{\tempLP\, \eta} - 2\,(S+Z)\,\gamma\,\eta}\right)\,.
    \end{equation}
    Let $\hat{\nu}^*$ be the approximation of $\nu^*$ (where we replace $h$ with $\hat{h}$); we obtain that
    \begin{equation}
        \begin{split}
            \hat{\nu}^*_t &= \frac{1}{4\, \eta\, \tempLP\, \gamma+ \frac{2}{\eta}}\left({\frac{2\,a_2\,\feeV\,(T-t)}{\tempLP\, \eta} - 2\,(S+Z)\,\gamma\,\eta}\right)\\
            &= \frac{2\,a_2\,\feeV\, (T-t) - 2\,(S+Z)\,\gamma\,\tempLP\, \eta^2}{4\, \eta^2\, \tempLP^2\, \gamma +2\, \tempLP}\,.
        \end{split}
    \end{equation}

\begin{rem}
For a risk-neutral venue, if the noise in the liquidity provision is small $(\eta\downarrow 0)$, we have that
\begin{equation}
    \hat{\nu}^*_t \to \frac{a_2\,\feeV\,(T-t)}{\tempLP}\,.
\end{equation}    
That is, the provision of liquidity is inversely proportional to the transaction costs in the external venue $\tempLP$, and directly proportional to (i) the fees collected by the venue $\feeV$ and (ii) the sensitivity of order flow to the depth of the pool $a_2$. In particular, if $a_2 = 0$, the LP does not add nor remove liquidity from the pool in equilibrium. 
\end{rem}

\subsubsection{Risk-averse venue}
Here we work out the approximate solution when the venue is risk-averse.
We take $\nu_{\infty}\uparrow \infty$ such that $\nu^*(A_t) = \AB_t/(2\,\tempLP\, \eta)$. We also employ the approximation $e^{-x}\approx 1 - x $ and we only consider the case where $Y_{t\shortminus}>\xi$. We obtain that
\begin{align}
    0 &\approx - \frac{1}{4} \frac{\left[\left( \frac{1}{\tempLP\, \eta} + 2\, \zeta\, \eta \right)\,\partial_Y v - 2\, \gamma\, \eta\, (S+Z)\right]^2}{2\,(\gamma+\zeta) + \frac{1}{\tempLP\, \eta^2}} - \frac{1}{2}\, \frac{\sigma^2\,(\zeta\,\partial_S v - \gamma\,Y)^2}{\gamma+\zeta}\\
    &\quad - \sum_{i\in\{+,-\}} \barlambdai(Z, Y, S)\,\left(\vi - v + \feeV + \deltai\,\xi\,\left(S - \frac{Z\,Y}{Y+\deltai\,\xi}\right)\right)\\
    &\quad + \frac{1}{2}\, \gamma\, \eta^2\, S^2 + \gamma\, \eta^2\, S\, Z + \frac{1}{2}\, \gamma\, \eta^2\, Z^2 + \frac{1}{2}\, \gamma\, \sigma^2\, Y^2\\
    &\quad + \frac{1}{2}\, \zeta\, \sigma^2\, (\partial_S v)^2 - \frac{1}{2}\, \sigma^2\, \partial_{SS} v + \frac{1}{2}\, \zeta\, \eta^2\, (\partial_Y v)^2 - \frac{1}{2}\,\eta^2\, \partial_{YY} v - \partial_t v\,,
\end{align}
where  $\vab=v(t,Z,Y+\deltaab\,\xi, S)$. \\

We focus on the regime where $\barlambdaab(t, Z_{t\shortminus}, Y_{t\shortminus}, S_t) > a_0$. Furthermore, we make the following approximations
\begin{align}
    \frac{Y}{Y+\deltaab\,\xi} &= \sum_{n=0}^\infty (-1)^n\, \left(\frac{\deltaab\,\xi}{Y}\right)^n = 1 + O\left(\frac{\xi}{Y}\right)\,,\\
    \frac{Y^2}{Y+\deltaab\,\xi} &= Y + \sum_{n=1}^\infty (-1)^n\, \left(\frac{(\deltaab\,\xi)^n}{Y^{n-1}}\right) = Y - \deltaab\,\xi + O\left(\frac{\xi^2}{Y}\right)\,,\\
    \frac{Y^2}{(Y+\deltaab\,\xi)^2} &= \sum_{n=0}^\infty (n+1)\, (-1)^n\, \left(\frac{\deltaab\,\xi}{Y}\right)^n = 1 + O\left(\frac{\xi}{Y}\right)\,,\\
    \frac{Y^3}{(Y+\deltaab\,\xi)^2} &= \sum_{n=0}^\infty Y\, (n+1)\, (-1)^n\, \left(\frac{\deltaab\,\xi}{Y}\right)^n = Y - 2\, \deltaab\, \xi + O\left(\frac{\xi^2}{Y}\right)\,,
\end{align}
which implies that 
\begin{equation}
    \begin{split}
        &v(t, Z\, Y^2/(Y+\deltaab\,\xi)^2, Y+\deltaab\,\xi, S) \\
        &\qquad = v(t, Z, Y+\deltaab\,\xi, S) + Z\,\left( \frac{Y^2}{(Y+\deltaab\,\xi)^2} - 1 \right)\, \partial_Z v(t,Z,Y+\deltaab\,\xi,S) + O\left(\frac{\xi}{Y}\right)\\
        &\qquad \approx v(t, Z, Y+\deltaab\,\xi, S)\,,
    \end{split}
\end{equation}
and 
\begin{equation}
    \begin{split}
        &Y\, v(t, \Zab, Y+\deltaab\,\xi, S) \\
        &\qquad = Y\, v(t, Z, Y+\deltaab\,\xi, S) + Z\,\left( \frac{Y^3}{(Y+\deltaab\,\xi)^2} - Y \right)\, \partial_Z v(t,Z,Y+\deltaab\,\xi,S) +  O\left(\frac{\xi}{Y}\right)\\
        &\qquad \approx Y \,v(t, Z, Y+\deltaab\,\xi, S) - 2\,\deltaab\, \xi\, Z\, \partial_Z v(t, Z, Y+\deltaab\,\xi, S)\,.
    \end{split}
\end{equation}

We obtain that 
\begin{equation}
    \begin{split}
        0 &\approx  - \frac{1}{4} \frac{\left[\left( \frac{1}{\tempLP\, \eta} + 2\, \zeta\, \eta \right)\,\partial_Y v - 2\, \gamma\, \eta\, (S+Z)\right]^2}{2\,(\gamma+\zeta) + \frac{1}{\tempLP\, \eta^2}} - \frac{1}{2}\, \frac{\sigma^2\,(\zeta\,\partial_S v - \gamma\,Y)^2}{\gamma+\zeta}\\
        &\quad - \sum_{i\in\{+,-\}} \left(a_1+a_2\,Y+a_3\,\deltai\,(Z-S)\right)\,(\vi - v)- 2\, a_1\, \feeV - 2\, a_2\, \feeV\, Y\\
        &\quad - 2\,a_2\, \xi^2\, Z + 2\, a_3\, \xi\, (S-Z)^2 + a_2\,\sum_{i\in\{+,-\}} \left[2\, \deltai\, \xi\, Z\, \partial_Z v(t, Z, Y+\deltai\,\xi, S) \right]\\
        &\quad + \frac{1}{2}\, \gamma\, \eta^2\, S^2 + \gamma\, \eta^2\, S\, Z + \frac{1}{2}\, \gamma\, \eta^2\, Z^2 + \frac{1}{2}\, \gamma\, \sigma^2\, Y^2\\
        &\quad + \frac{1}{2}\, \zeta\, \sigma^2\, (\partial_S v)^2 - \frac{1}{2}\, \sigma^2\, \partial_{SS} v + \frac{1}{2}\, \zeta\, \eta^2\, (\partial_Y v)^2 - \frac{1}{2}\,\eta^2\, \partial_{YY} v - \partial_t v\,.
    \end{split}
\end{equation}
We let $\hat v$ be the approximation of $v$ such that the sign $\approx$ above turns into an equality, and employ the ansatz
$$\hat v(t, Z, Y, S) = g_{11}(t) + 2\, \statevn^\top\, G_1(t) + \statevn^\top\, G_2(t)\, \statevn\,,$$
where $\statevn =  \begin{pmatrix} Z & Y & S\end{pmatrix}^\top$ and
\begin{align}
    G_1(t)&= \begin{pmatrix}
        g_{12}(t) & g_{13}(t) & g_{14}(t)
    \end{pmatrix}^\top\,,\\
    G_2(t) &= \begin{pmatrix}
        g_{22}(t) & g_{23}(t) & g_{24}(t)\\
        g_{23}(t) & g_{33}(t) & g_{34}(t)\\
        g_{24}(t) & g_{34}(t) & g_{44}(t)
    \end{pmatrix}\,,
\end{align}
with terminal conditions $g_{ij}(T)=0$ for $i,j\in\{1,2,3,4\}$.\\

We obtain the following system of differential equations:
\begin{align}
    0 &= -G_2{'}(t) + G_2(t)\, U\, G_2(t) + V^\top\, G_2(t) + G_2(t)\, V + R\,,\\
    0 &= G_1{'}(t) + C(t)\,G_1(t)^\top + E(t)\,,\\
    0 &= g_{11}{'}(t)+2\,a_1\,\feeV + \frac{\left(1+2\,\tempLP\,\zeta\,\eta^2-4\,\tempLP^2\,\gamma\,\zeta\,\eta^4\right)\,g_{13}(t)^2}{\tempLP\,(1+2\,\tempLP\,(\gamma+\zeta)\,\eta^2}\\
    \quad &\qquad -\frac{2\,\gamma\,\zeta\,\sigma^2\,g_{14}(t)^2}{\gamma+\zeta}+(2\,a_1\,\xi^2+\eta^2)\,g_{33}(t)+\sigma^2\,g_{44}(t)\,,
\end{align}
where for the first equation, $U, V$, and $R$ are given by
\begin{align}
    U &=
    \begin{pmatrix}
        0 & 0 & 0\\
        0 & -\frac{1+2\,\tempLP\,\zeta\,\eta^2-4\,\tempLP^2\,\gamma\,\zeta\,\eta^4}{\tempLP\,(1+2\,\tempLP\,(\gamma+\zeta)\,\eta^2)} & 0\\
        0 & 0 & \frac{2\,\gamma\,\zeta\,\sigma^2}{\gamma+\zeta}
    \end{pmatrix}\,,\\
    V &= \begin{pmatrix}
        0 & 0 & 0\\
        -2\,a_3\,\xi+\frac{\gamma\,\eta^2\,(1+2\,\tempLP\,\zeta\,\eta^2)}{1+2\,\tempLP\,(\gamma+\zeta)\,\eta^2} & 0 & 2\,a_3\,\xi+\frac{\gamma\,\eta^2\,(1+2\,\tempLP\,\zeta\,\eta^2)}{1+2\,\tempLP\,(\gamma+\zeta)\,\eta^2} \\
        0 & \frac{\gamma\,\zeta\,\sigma^2}{\gamma+\zeta} & 0
    \end{pmatrix}\,,\\
    R &= \begin{pmatrix}
        2\,a_3 + \frac{\gamma\,\eta^2\,(1+2\,\tempLP\,\zeta\,\eta^2)}{2+4\,\tempLP\,(\gamma+\zeta)\,\eta^2} & 0 & -2\,a_3 + \frac{\gamma\,\eta^2\,(1+2\,\tempLP\,\zeta\,\eta^2)}{2+4\,\tempLP\,(\gamma+\zeta)\,\eta^2}\\
        0 & \frac{\gamma\,\zeta\,\sigma^2}{2(\gamma+\zeta)} & 0\\
        -2\,a_3 + \frac{\gamma\,\eta^2\,(1+2\,\tempLP\,\zeta\,\eta^2)}{2+4\,\tempLP\,(\gamma+\zeta)\,\eta^2} & 0 & 2\,a_3 + \frac{\gamma\,\eta^2\,(1+2\,\tempLP\,\zeta\,\eta^2)}{2+4\,\tempLP\,(\gamma+\zeta)\,\eta^2}
    \end{pmatrix}\,,
\end{align}    
and for the second equation, $C(t)$ and $E(t)$ are 
\begin{align}
    C(t) &= \begin{pmatrix}
            0 & \frac{-\tempLP\,\gamma\,\eta^2\,(1+2\,\tempLP\,\zeta\,\eta^2)+2\,a_3\,\tempLP\,\xi\,(1+2\,\tempLP\,(\gamma+\zeta)\,\eta^2)+(1+2\,\tempLP\,\zeta\,\eta^2-4\,\tempLP^2\,\gamma\,\zeta\,\eta^4)\,g_{23}(t)}{\tempLP\,(1+2\,\tempLP\,(\gamma+\zeta)\,\eta^2)} & -\frac{2\,\gamma\,\zeta\,\sigma^2\,g_{24}(t)}{\gamma+\zeta}\vspace{0.2cm}\\
        0 & \frac{(1+2\,\tempLP\,\zeta\,\eta^2-4\,\tempLP^2\,\gamma\,\zeta\,\eta^4)\,g_{33}(t)}{\tempLP\,(1+2\,\tempLP\,(\gamma+\zeta)\,\eta^2)} & -\frac{\gamma\,\zeta\,\sigma^2\,(1+2\,g_{24}(t))}{\gamma+\zeta}\vspace{0.2cm}\\
        0 & \frac{-\tempLP\,(\gamma\,\eta^2\,(1+2\,\tempLP\,\zeta\,\eta^2)+2\,a_3\,\xi\,(1+2\,\tempLP\,(\gamma+\zeta)\,\eta^2))+(1+2\,\tempLP\,\zeta\,\eta^2-4\,\tempLP^2\,\gamma\,\zeta\,\eta^4)\,g_{34}(t)}{\tempLP\,(1+2\,\tempLP\,(\gamma+\zeta)\,\eta^2)} & -\frac{2\,\gamma\,\zeta\,\sigma^2\,g_{44}(t)}{\gamma+\zeta}
    \end{pmatrix},\\
    E(t) &= \begin{pmatrix}
        a_2\, \xi^2\,(1-4\,g_{23}(t)) & a_2\, (\feeV + \xi^2\,g_{33}(t))& 0
    \end{pmatrix}^\top.
\end{align}

The equation for $G_2$ is a matrix Riccati equation. Let us introduce the matrices $C=I_3$ and $D=0$, where $I_3$ denotes the identity matrix in $ \mathbb R^{3\times 3}$ and $0$ denotes the zero matrix in the same space. The terminal condition of the equation is $G_2(T)=0$, and it is therefore  clear that
$C + DG_2(T)+ G_2(T)D=C >0$.
We now introduce the matrix $\Theta$ given by
$$\Theta = \begin{pmatrix} -CV & -CU-V^\top D + D V^\top \\ 0 & -U^\top D
\end{pmatrix}.$$

After some calculations we see that $\Theta + \Theta^\top$ is equal to\\
\scalebox{0.85}{%
$
\begin{pmatrix}0 & 2 a_{3} \xi - \frac{\eta^{2} \gamma \left(2 \tempLP \eta^{2} \zeta + 1\right)}{2 \tempLP \eta^{2} \left(\gamma + \zeta\right) + 1} & 0 & 0 & 0 & 0\\2 a_{3} \xi - \frac{\eta^{2} \gamma \left(2 \tempLP \eta^{2} \zeta + 1\right)}{2 \tempLP \eta^{2} \left(\gamma + \zeta\right) + 1} & 0 & - 2 a_{3} \xi - \frac{\eta^{2} \gamma \left(2 \tempLP \eta^{2} \zeta + 1\right)}{2 \tempLP \eta^{2} \left(\gamma + \zeta\right) + 1} - \frac{\gamma s^{2} \zeta}{\gamma + \zeta} & 0 & - \frac{4 a^{2} \eta^{4} \gamma \zeta - 2 \tempLP \eta^{2} \zeta - 1}{a \left(2 \tempLP \eta^{2} \left(\gamma + \zeta\right) + 1\right)} & 0\\0 & - 2 a_{3} \xi - \frac{\eta^{2} \gamma \left(2 \tempLP \eta^{2} \zeta + 1\right)}{2 \tempLP \eta^{2} \left(\gamma + \zeta\right) + 1} - \frac{\gamma s^{2} \zeta}{\gamma + \zeta} & 0 & 0 & 0 & - \frac{2 \gamma s^{2} \zeta}{\gamma + \zeta}\\0 & 0 & 0 & 0 & 0 & 0\\0 & - \frac{4 a^{2} \eta^{4} \gamma \zeta - 2 \tempLP \eta^{2} \zeta - 1}{a \left(2 \tempLP \eta^{2} \left(\gamma + \zeta\right) + 1\right)} & 0 & 0 & 0 & 0\\0 & 0 & - \frac{2 \gamma s^{2} \zeta}{\gamma + \zeta} & 0 & 0 & 0\end{pmatrix}.$
}\newline

It is then clear that all the leading principal minors of the above matrix are non-positive; in particular, they are all equal to zero, except for the determinant of 
$$
\begin{pmatrix}
    0 & 2 a_{3} \xi - \frac{\eta^{2} \gamma \left(2 \tempLP \eta^{2} \zeta + 1\right)}{2 \tempLP \eta^{2} \left(\gamma + \zeta\right) + 1} \\
    2 a_{3} \xi - \frac{\eta^{2} \gamma \left(2 \tempLP \eta^{2} \zeta + 1\right)}{2 \tempLP \eta^{2} \left(\gamma + \zeta\right) + 1} & 0
\end{pmatrix}
$$
which is given by
$$- \left(2 a_{3} \xi - \frac{\eta^{2} \gamma \left(2 \tempLP \eta^{2} \zeta + 1\right)}{2 \tempLP \eta^{2} \left(\gamma + \zeta\right) + 1} \right)^2 \le 0.$$
Therefore, the matrix $\Theta + \Theta^\top$ is negative semi-definite, and Theorem 3.6.6 in \cite{abou2012matrix} implies that there exists a well-defined solution of the Riccati equation on $(-\infty, T]$.

\section{Numerical results}\label{sec:numerical-results}
For the experiments below, we  discretize $[0,T]$ in 10,000 timesteps with $T=1$ (one day). We employ market data from Binance and Uniswap V2 in the pair ETH-USDC between 1 January 2022 and 30 April 2022 to calibrate model parameters. The initial exchange rate is $S_0 = Z_0 = 2820$ in the ETH-USDC pair and the daily volatility in Binance is $\sigma = 0.0569 \times S_0$.\footnote{We obtained the volatility estimate from \href{https://marketmilk.babypips.com/symbols/ETHUSD/volatility}{market-milk daily volatility}.} The average traded quantity in a ten minute window is $\xi = 300$, and the fee collected by the venue per transaction is approximated by the constant fee 
$\feeV = 0.01 \times \xi\times Z_0$. Here, each jump in $\Nab$ captures the aggregate trading that happened in a ten-minute window. 
The initial position in the pool is $Y_0 = 50,000$ ETH,\footnote{This value is motivated by those in \href{https://app.uniswap.org/explore/pools/ethereum/0x88e6A0c2dDD26FEEb64F039a2c41296FcB3f5640}{Uniswap's ETH-USDC pool} ---accessed 20 February 2025.  } and the depth of the pool is $c_0 = Y_0\,\times (Y_0 \times Z_0)$.
We use $\eta= 10^{-10}$ ETH as the volatility of the Brownian motion used in the provision of liquidity.
For $t\in\mfT$, we estimate $\lambdaab_t$ using data from the previous ten minutes. We then fit a linear regression to calibrate $a_1,a_3$ in \eqref{eq:intensities}; for now, we set $a_2 = 0$. We obtain the estimates $a_1= 142.7$ and $a_3= 13.6$. These values were scaled so that a jump corresponds to the activity in a ten-minute window in line with $\xi$. Below, we study the strategies for the  risk-averse case. The LP's risk aversion parameter is $\gamma =  10^{-18}$ and the venue's risk aversion is $\zeta=10^{-6}$. Lastly, for the figures 
below we let the temporary price impact parameter in Binance be $\tempLP = 10^{-14}$ (negligible walking the book costs in the external venue). Below, we show that as we increase the value of this parameter, the strategy collapses around zero and the representative LP does not add liquidity to the pool.\\

Figure \ref{fig:violations-simulations} shows the violations to the approximation assumption within the simulation. That is, the violations to the assumption that both $\lambdaa,\lambdab>0$. We observe that there are no violations to the non-negative assumption in the simulations.

\begin{figure}[H]
    \centering    \includegraphics[width=0.5\linewidth]{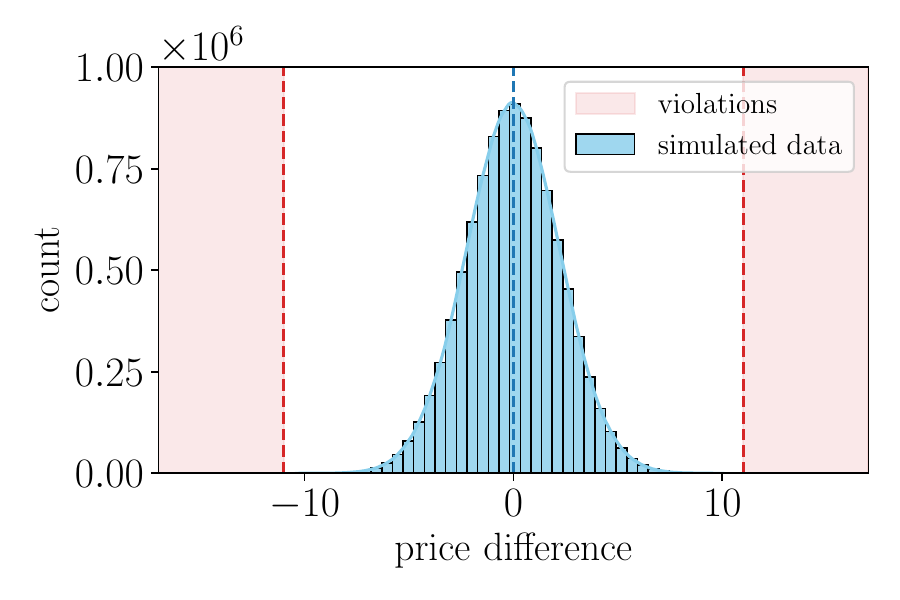}
    \caption{Histogram of simulated price differences between Binance and Uniswap V2 for ETH-USDC using 1,000 simulations and the model parameters at the start of the section. The red shaded area represents the region in which $a_1\pm\,a_3 (S_t-Z_t)$ becomes negative. We take $a_2=0$ for simplicity so that the violation boundary is fixed and does not depend on the number of ETH units in the pool. }
    \label{fig:violations-simulations}
\end{figure}

Figure \ref{fig:sample-path-and-regions} shows a sample path for the inventory $Y_t$, together with the prices $S_t$ (external) and $Z_t$ (pool).

\begin{figure}[H]
    \centering
    \includegraphics[width=0.4\linewidth]{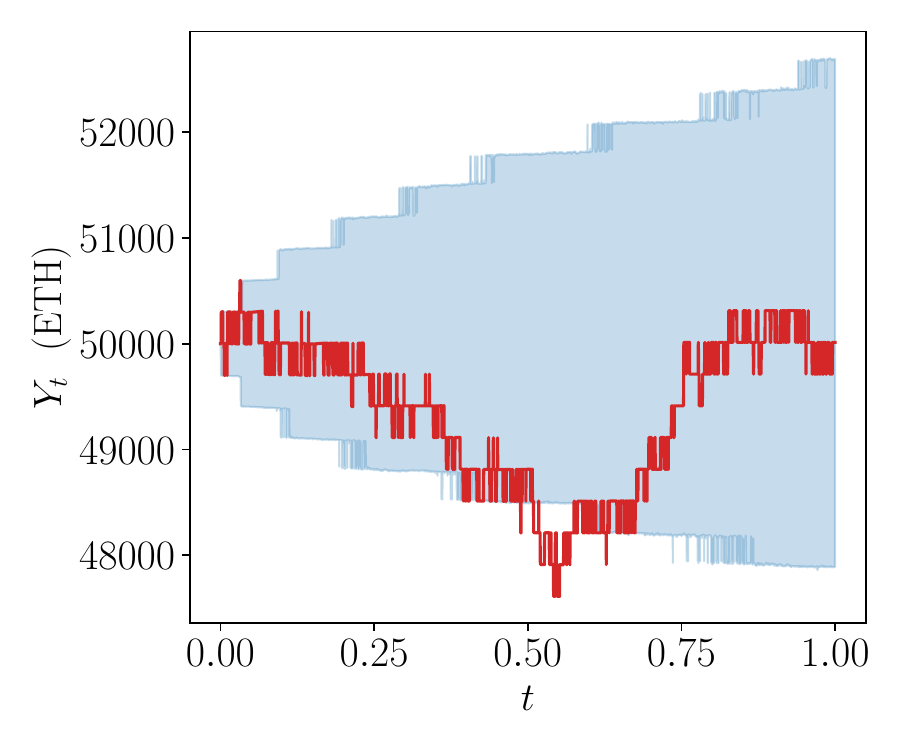}
    \includegraphics[width=0.4\linewidth]{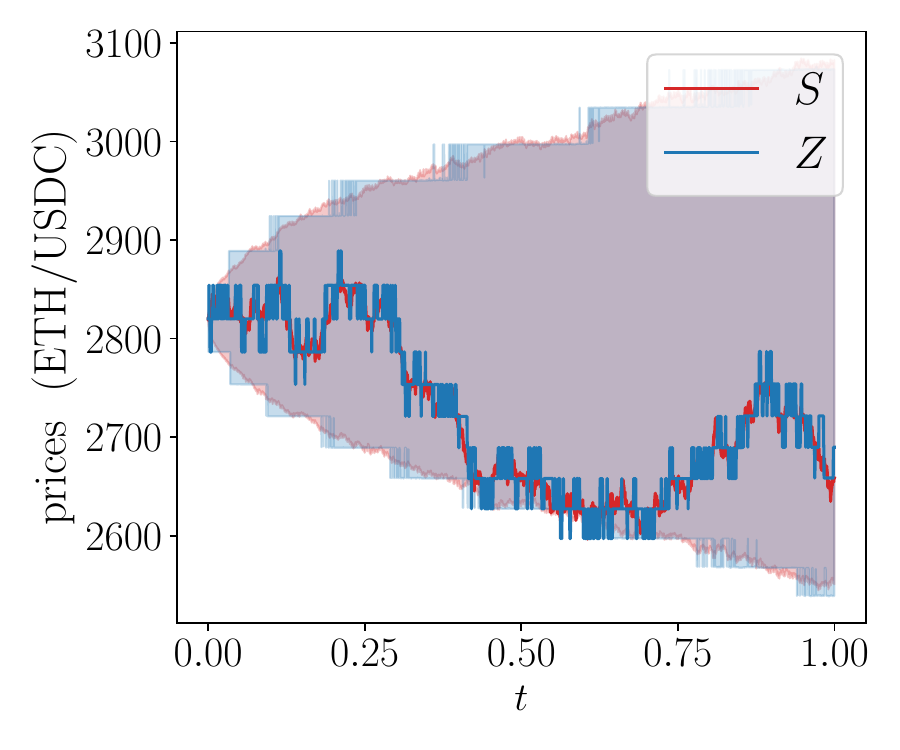}
    \caption{Sample path (with 90\% bands across time) for the inventory of ETH in the pool (left panel), and the instantaneous exchange rate in the pool and outside (right panel). }
    \label{fig:sample-path-and-regions}
\end{figure}

As expected, arbitrageurs keep the prices in the pool aligned to those outside of the pool. For these parameters, with $a_2=0$, the changes in $Y$ are relatively small and seem to be driven mostly by the liquidity taking activity;  we confirm this below. Figure \ref{fig:sample-path-and-regions-2} shows   the speed $\nu_t$ at which the LP adds or subtracts liquidity from the pool, and the cumulative liquidity provided (given by $\int_0^t \nu_s\,\d s$).

\begin{figure}[H]
    \centering
    \includegraphics[width=0.4\linewidth]{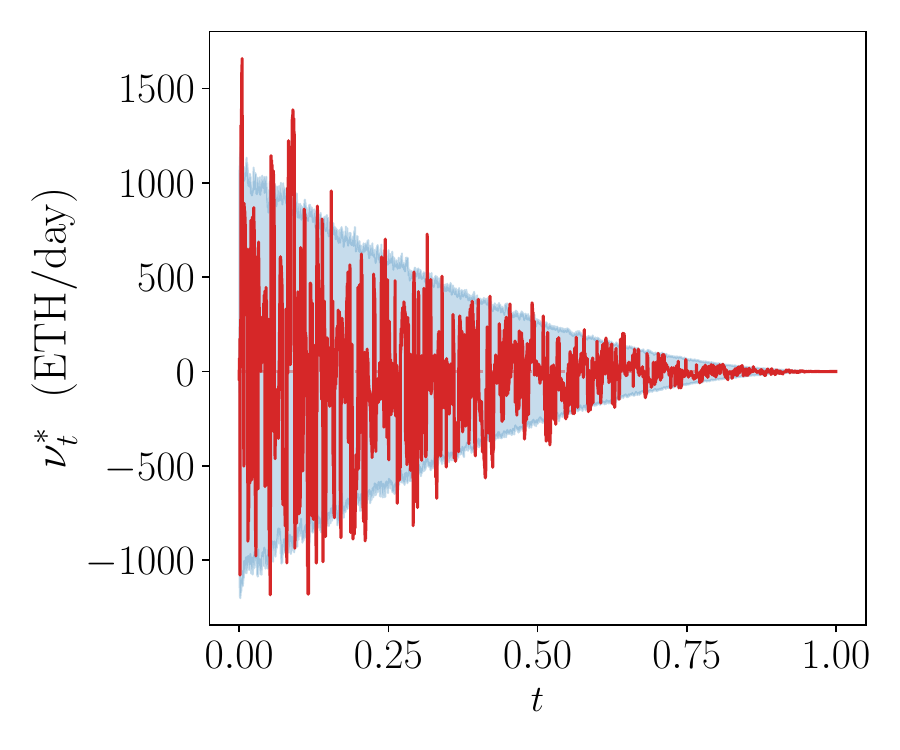}
    \includegraphics[width=0.4\linewidth]{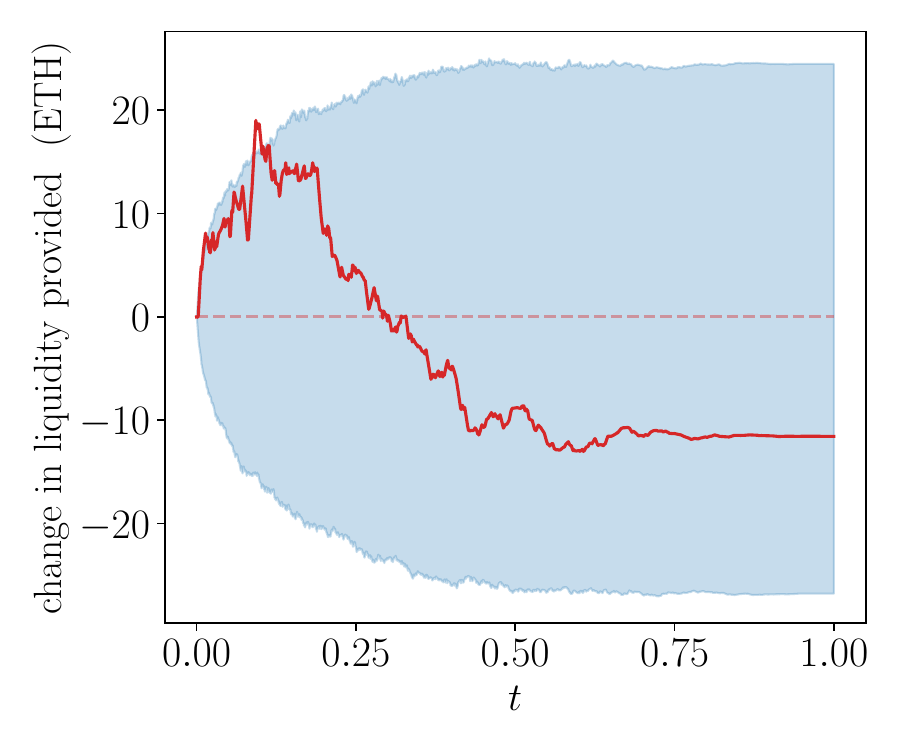}
    \caption{Sample path (with 90\% bands across time) for the  speed at which the LP adds/removes liquidity from the pool (left panel), and the cumulative change in liquidity provided (right panels), given by $\int_0^t \nu_s\,\d s$. }
    \label{fig:sample-path-and-regions-2}
\end{figure}

Next, we study the optimal strategy $\nu^*_t$ of the LP in more detail.
Figure \ref{fig:stressing-parameters-nustar} shows a sample path of the optimal strategy as we stress the values of model parameters (the randomness is the same for all simulations).

\begin{figure}[H]
    \centering
    \includegraphics[width=0.38\linewidth]{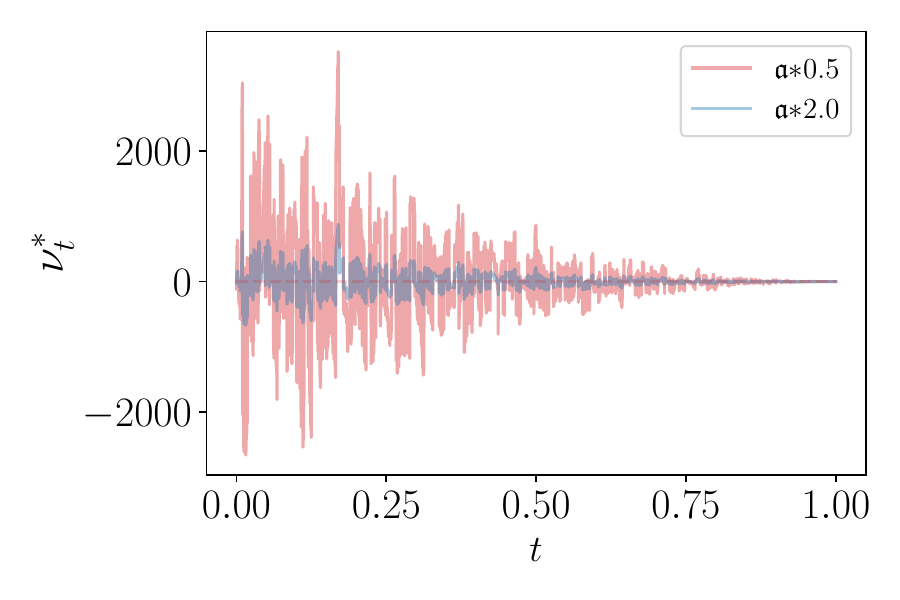}
    \includegraphics[width=0.38\linewidth]{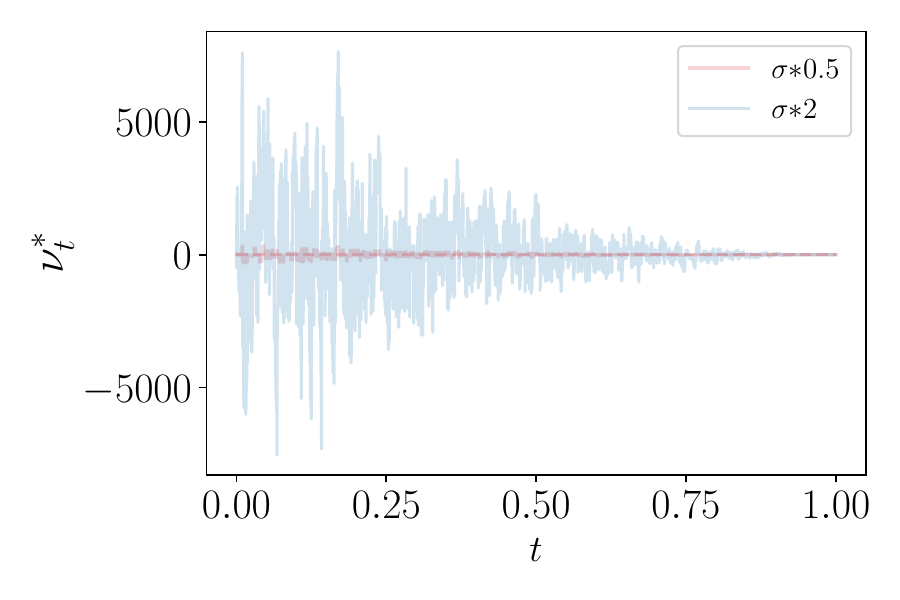}
    \caption{Sample path of the optimal strategy of the LP as model parameters $\tempLP$ and $\sigma$ change.}
    \label{fig:stressing-parameters-nustar}
\end{figure}

As expected, the higher the volatility, the higher the variability of the strategy of the the representative LP.  The smaller the temporary price impact in the external venue $\tempLP$, the more trading activity the representative LP can afford. When we increase the transaction cost parameter  for trading in the external venue to $\tempLP\in\{10^{-13}, 10^{-12}\}$, we obtain that the strategy of the representative LP collapses around zero, as shown in Figure \ref{fig:sample-path-and-regions-3}.\\

\begin{figure}[H]
    \centering
    \includegraphics[width=0.38\linewidth]{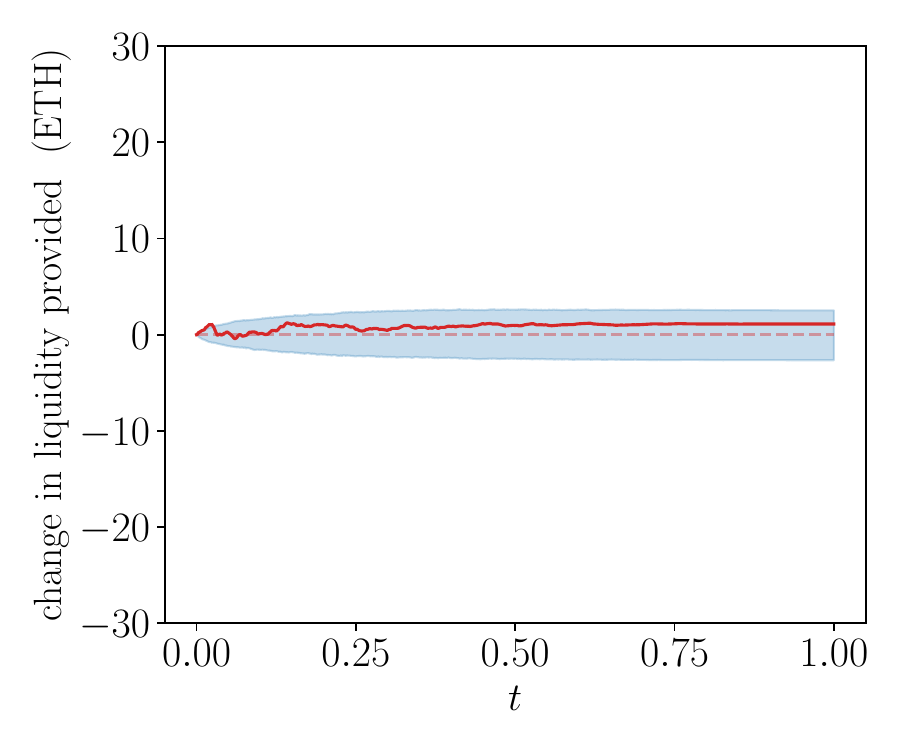}
    \includegraphics[width=0.38\linewidth]{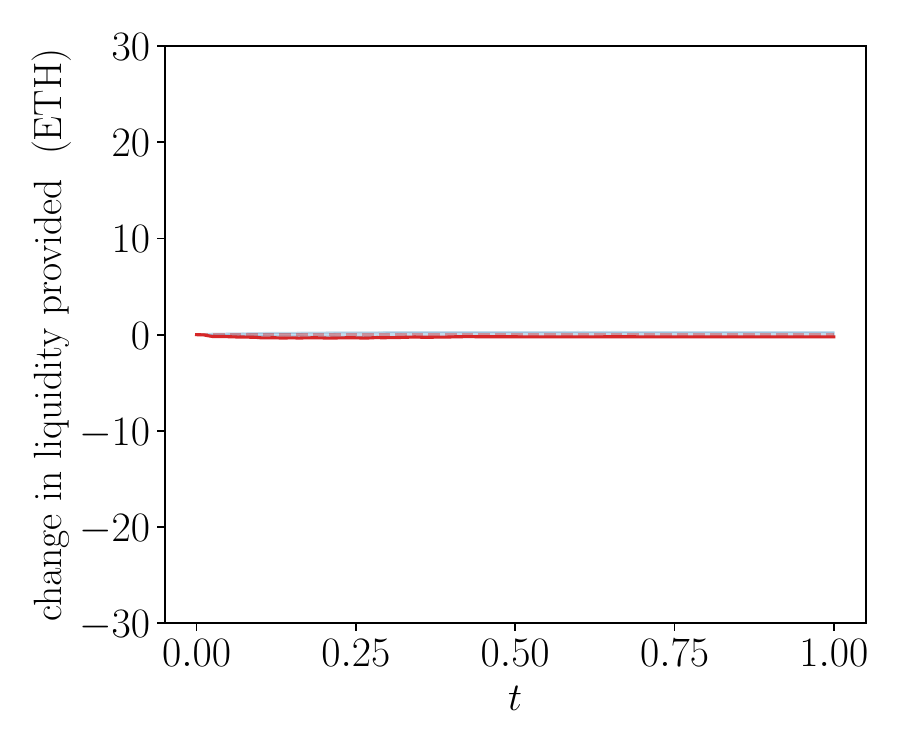}
    \caption{Sample path (with 90\% bands across time) for the  cumulative change in liquidity provided given by $\int_0^t \nu_s\,\d s$. The first  plot is for $\tempLP=10^{-13}$ and the second plots is for $\tempLP=10^{-12}$.  }
    \label{fig:sample-path-and-regions-3}
\end{figure}

\subsection{Distributing the fees collected}
The starting point of the contract representation, i.e., $P_0$, shifts the profitability of the LP to the right (the higher $P_0$ the more profitable), and it shifts the profitability of the venue to the left (the higher $P_0$ the less profitable). Here, we take $P_0$ such that both the venue and the representative LP have the same average wealth.\\

Figure \ref{fig:rew_venue} shows  $10,000$ Monte-Carlo simulations of the the equilibrium of the model. We report the distribution of the final reward $\rew$ and the venue's PnL. The red line is the mean of the distribution. On both sides, the mean is positive. The higher the right aversion of the venue, the more concentrated these profits would be.

\begin{figure}[H]
    \centering
    \includegraphics[width=0.4\linewidth]{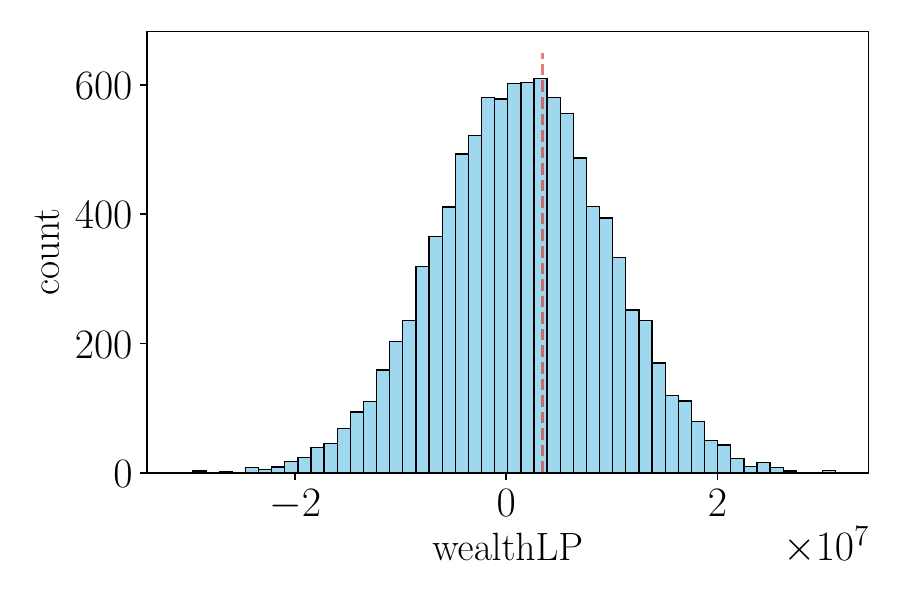}
    \includegraphics[width=0.4\linewidth]{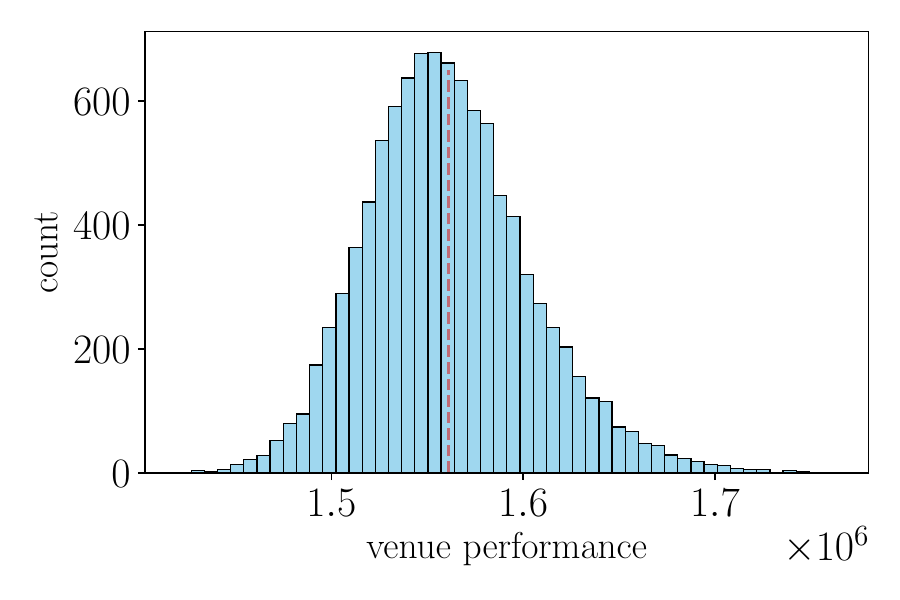}
    \caption{Histograms for the reward $\rew$ paid to the representative LP (left panel) and the performance  of the venue $\feeV\,(\Na_T + \Nb_T) - \rew$ (right panel). The mean values are in vertical dotted lines. }
    \label{fig:rew_venue}
\end{figure}

The optimal contract benefits both the LPs and the venue. This can be seen in the positive values of the means of both histograms.

\subsection{Attracting noise trading}
In the analysis we carried out so far, the  intensity of order arrivals does not react to the quantity of asset $Y$ in the pool, i.e.  $a_2 = 0$. As shown above, this implies that the 
strategy of the representative LP collapses around zero for meaningful values of the transaction cost parameter for trading in the external venue. 
Here, we show that this is no longer the case if $a_2>0$. In what follows, we use $a_2 = 10^{-5}$ and $\tempLP = 5\times 10^{-6}$. Note that this value of $a_2$ only represents around  $0.5$ ETH/10-minutes of added intensity (given that the amount of ETH in the pool  is around $50,000$). \\

Figure \ref{fig:real-1} shows a trajectory with their confidence bands of the quantity of ETH in the pool, the optimal speed at which the representative LP adds liquidity to the pool, the cumulative liquidity provided, and the cumulative fees paid to the external venue.

\begin{figure}[H]
    \centering
    \includegraphics[width=0.4\linewidth]{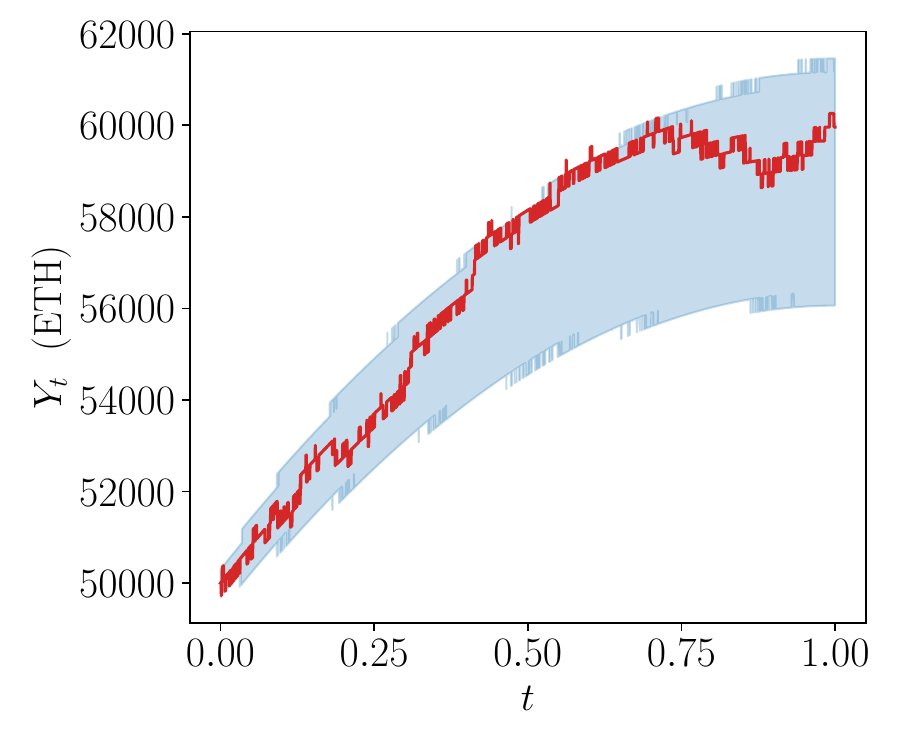}
    \includegraphics[width=0.4\linewidth]{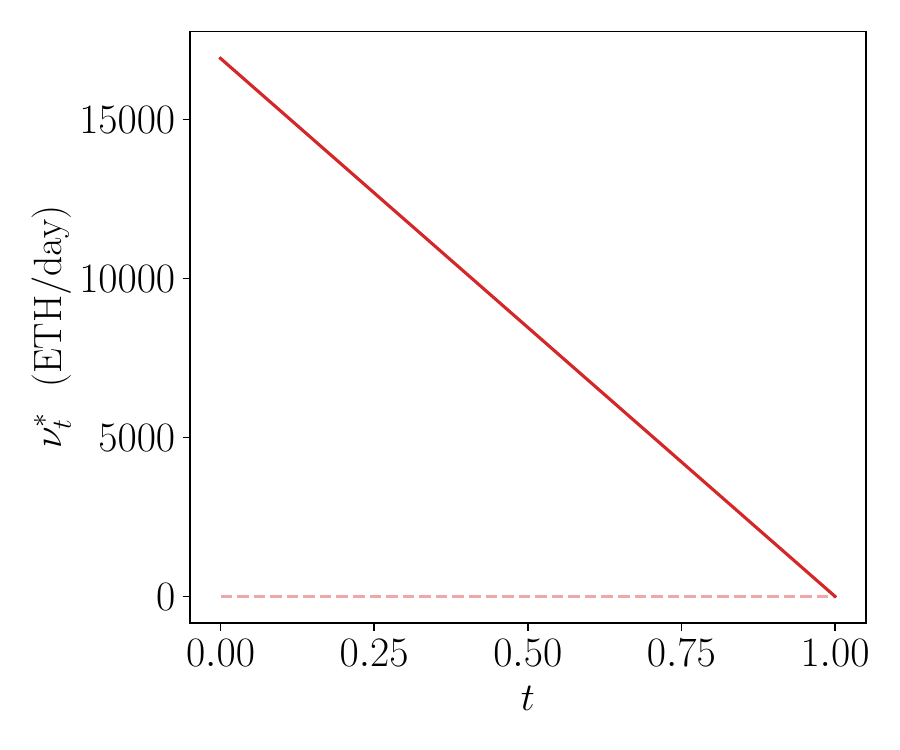}
    \includegraphics[width=0.4\linewidth]{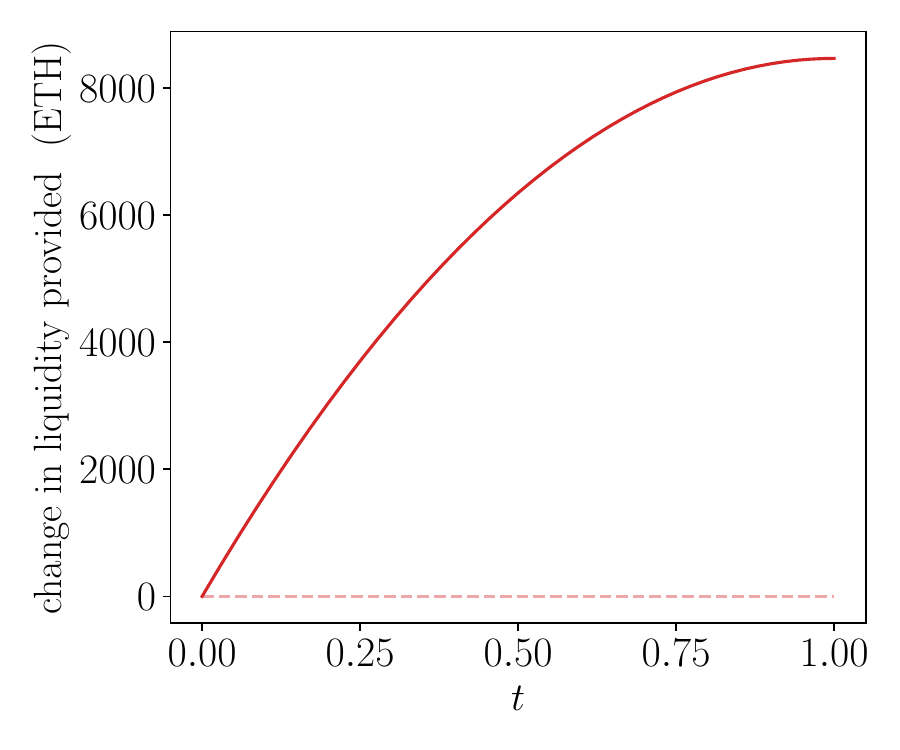}
    \includegraphics[width=0.4\linewidth]{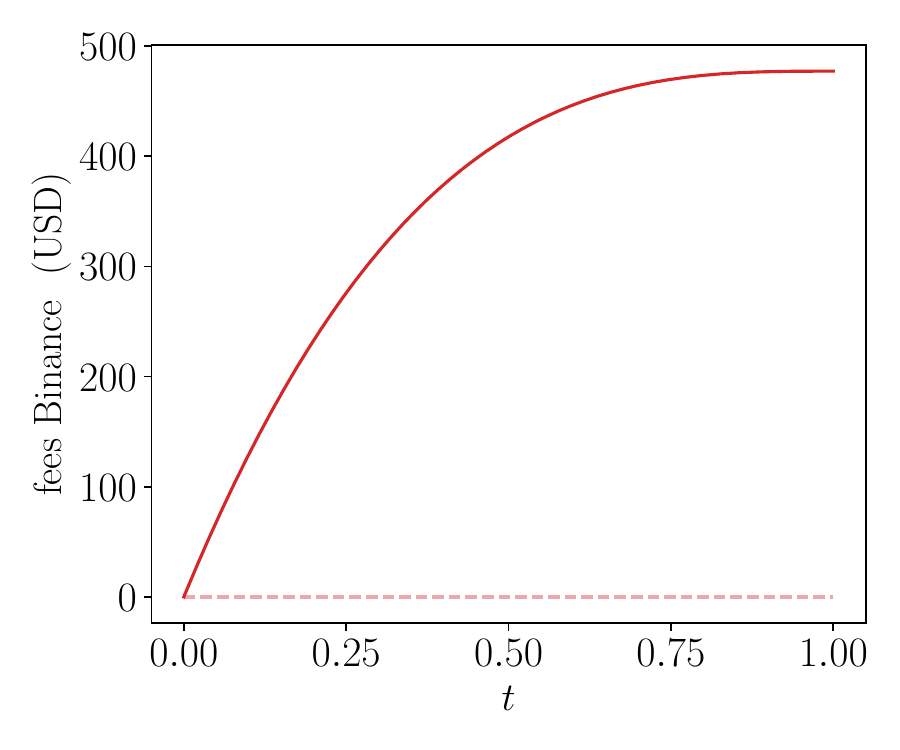}
    \caption{Sample path (together with 90\% confidence bands) of the (i) quantity of ETH in the pool (top left panel), (ii) the optimal speed at which the representative LP adds liquidity to the pool (top right panel), (iii) the cumulative liquidity provided (bottom left panel), and (iv) the cumulative fees paid to the external venue, given by $\int_0^t \tempLP\,(\nu^{*}_u)^2\,\d u$ (bottom right panel). }
    \label{fig:real-1}
\end{figure}

With these model parameters, the representative LP adds roughly 8000 ETH to the pool for which she pays 500 USDT in fees; for comparison, the spread in Binance for ETH-USDT is typically 0.01 which implies that the trading costs in Binance are in the correct order of magnitude. Unlike the case in the previous section, to the naked eye, there is no variability in the trading speed. The confidence bands in the amount of ETH in the pool is due to the trading activity of liquidity takers. 
Next, we study the effect of $\feeV$ and $\tempLP$ in the  optimal trading strategy of the representative LP in Figure \ref{fig:real-2}. \\

\begin{figure}[H]
    \centering
    \includegraphics[width=0.4\linewidth]{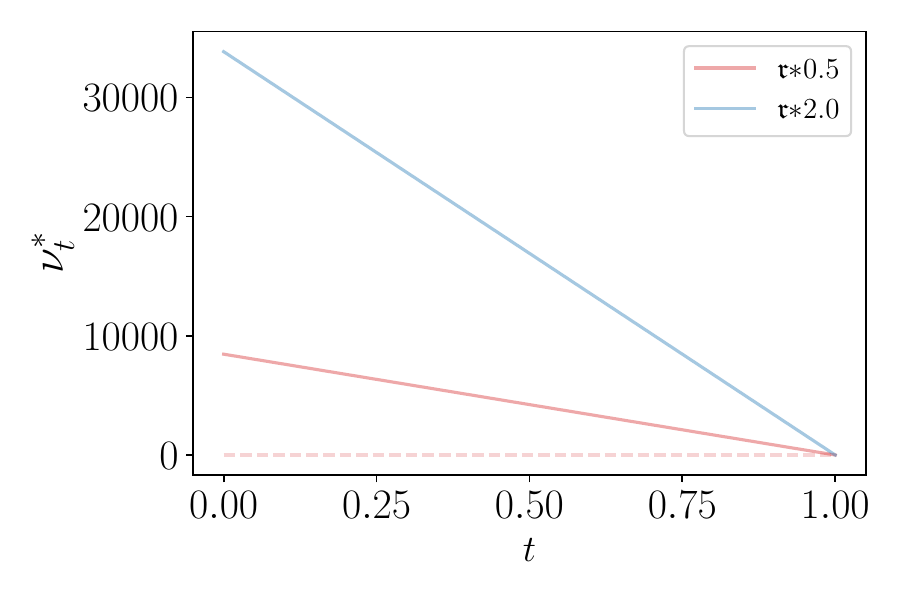}
    \includegraphics[width=0.4\linewidth]{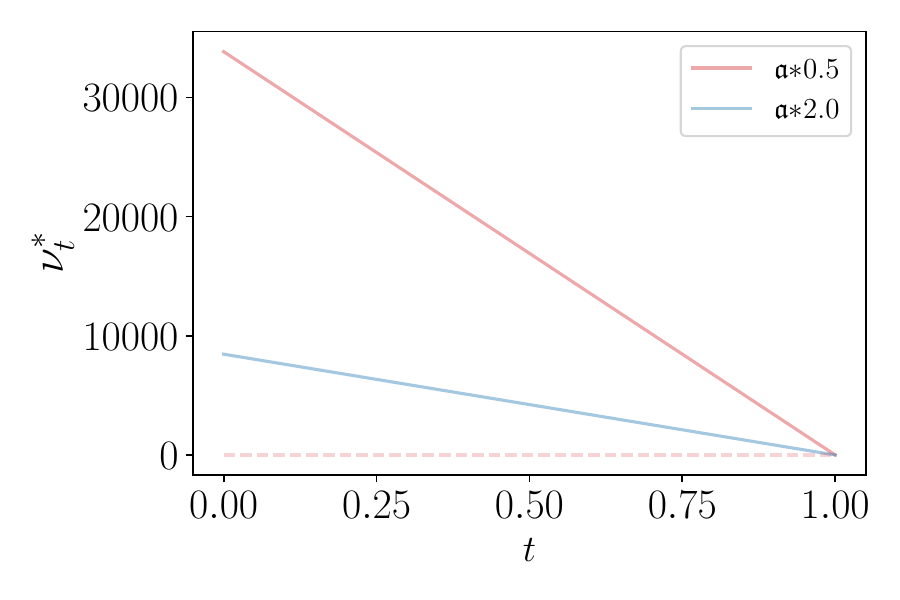}
    \caption{Sample path of the optimal strategy of the LP as model parameters $\feeV$ and $\tempLP$ change. }
    \label{fig:real-2}
\end{figure}

As expected, all else being equal,  the higher the fee charged to LTs, the more the venue collects and the better the contract that the venue offers the representative LP, which translates in more liquidity added to the pool. The effect of the transaction cost parameter is similar to that in the previous section. Higher transaction costs in the external venue diminish the activity in the pool.\\

Finally, we repeat the analysis in Figure \ref{fig:rew_venue} for this new set of model parameters in Figure \ref{fig:real-3}.

\begin{figure}[H]
    \centering
    \includegraphics[width=0.4\linewidth]{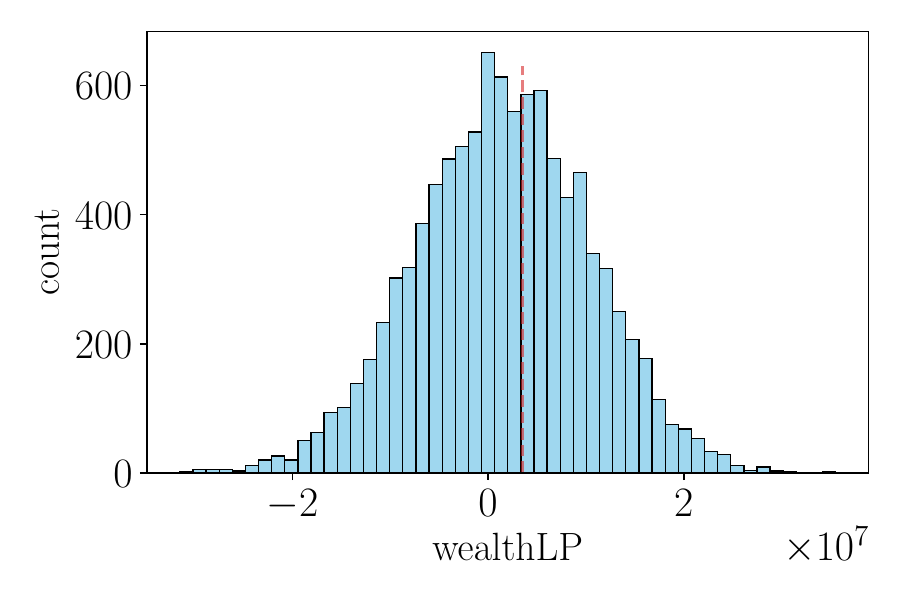}
    \includegraphics[width=0.4\linewidth]{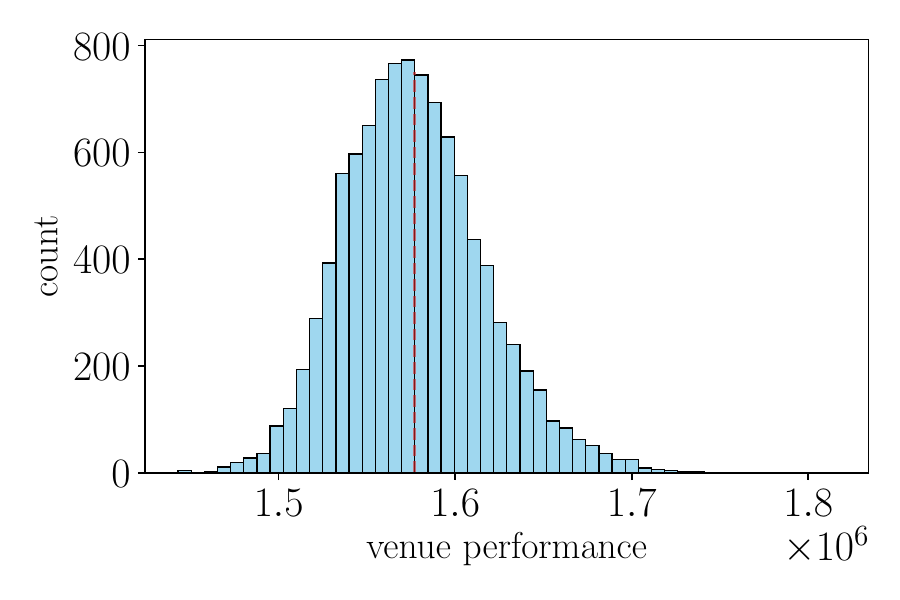}
    \caption{Histograms for the reward $\rew$ paid to the representative LP (left panel) and the performance  of the venue $\feeV\,(\Na_T + \Nb_T) - \rew$ (right panel). The mean values are in vertical dotted lines. }
    \label{fig:real-3}
\end{figure}

We see that the performance of both players is still positive in line with the results in Figure \ref{fig:rew_venue}.
All else being equal,  when $a_2 >0$, both buy order arrivals and sell order arrivals increase with the liquidity in the pool, and this additional order flow is not directional (noise trading). The representative LP exploits this by adding liquidity to the pool and attracting more of these type of trades, which in turn allows the venue to offer a higher reward to the representative LP. We remark that for the above histogram, the transaction cost parameter for trading in the external venue is much higher than the one used in Figure \ref{fig:rew_venue}.

\section{Conclusion}\label{conclusion}
We characterized the Stackelberg equilibrium of a venue and the representative LP. Trading took place in an AMM and while the representative LP aims to maximize profitability, the venue wishes to collect as much fees and order flow as possible. We find that if higher liquidity does not attract noise trading, even when the venue employs the optimal contract, the representative LP does not have incentives to add liquidity to the pool. 
On the other hand, if higher liquidity in the AMM attracts noise trading, then the representative LP adds as much liquidity as the transaction costs in the external venue allow. 
Using data from 1 January 2022 to 30 April 2022, we find a positive correlation of 11\% between the daily average depth in $Y$ for the ETH-USDC pool and the daily order flow, furthermore, the slope of the linear regression is positive but there is no evidence to reject the hypothesis that the slope is equal to zero (the $p$-value is 0.232).
Our work provides insights into the optimal design of these venues going forward and the testable conditions for the well functioning of these venues.

\section*{Acknowledgments:}
We are grateful to the Oxford-Man Institute of Quantitative Finance and the Fintech Chair at Université Paris-Dauphine for support. We thank Fayçal Drissi for providing us with data. \\

\noindent For the purpose of open access, the authors have applied a CC BY public copyright license to any author accepted manuscript arising from this submission.

\bibliographystyle{apalike}
\bibliography{References}

\newpage
\appendix

\section{Proof of Theorem \ref{contractrep} and \ref{theorem: contract rep value}}\label{proofcontractrep}

\subsection{Lemmas}

For any stopping time $\tau$ valued in $[0,T]$, we define
\begin{equation}
    J(\tau, \nu) :=  \E^\nu_\tau \left[ - \exp \left\{ -\gamma\, \left(\rew + \rwealth_{\tau,T}^\nu\right) \right\} \right]\,,
\end{equation}
where $\nu$ is in $\mcA_{\tau,T}$ (the restriction of $\mcA$ on time horizon $[\tau,T]$) and $\rwealth$ is the running wealth
\begin{align}
    \rwealth_{s,t}^\nu &= - \int_s^t \left\{\tempLP\, \nu_{u\shortminus}^2 \right\}\, \d u + \eta\, \int_s^t (S_{u\shortminus} + Z_{u\shortminus})\, \d B_u^{\nu}\\
    &\phantom{{}={}} + \sigma\, \int_s^t Y_{u\shortminus}\, \d W_u + \sum_{i\in\{+,-\}} \int_s^t \left\{ \deltai\, \xi\, \left( S_{u} - \frac{Z_{u\shortminus}\, Y_{u\shortminus}}{Y_{u\shortminus} + \deltai\, \xi} \right) \right\}\, \d \Ni_u\,.
\end{align}
Recall that the value function $V$ is defined as
\begin{equation}
    V_\tau = \esssup_{\nu\in \mcA_{\tau,T}} J(\tau, \nu)\,.
\end{equation}

\begin{lemma}\label{DPPMM}
    Let $\tau$ be a stopping time valued in $[t,T]$. Then
    \begin{equation}
        V_t = \esssup_{\nu\in \mcA_{t,\tau}} \E^\nu_t \left[ \exp\{-\gamma\, \rwealth_{t,\tau}^\nu\}\, V_\tau\right]\,,
    \end{equation}
\end{lemma}

\begin{proof}
    By tower property and the definition of the value function $V_t$, we have that
    \begin{align}
        V_t &= \esssup_{\nu \in \mcA_{t,\tau}} \E_t^\nu \left[ \exp \left\{ - \gamma\, \rwealth^\nu_{t,\tau} \right\}\, \E^\nu_\tau \left[ - \exp \left\{ -\gamma\, \left(\rew + \rwealth_{\tau,T}^\nu\right) \right\} \right]  \right]\\
        &\leq \esssup_{\nu \in \mcA_{t,\tau}} \E_t^\nu \left[ \exp \left\{ - \gamma\, \rwealth^\nu_{t,\tau} \right\}\, V_\tau  \right]\,.
    \end{align}
    For the other direction, by Proposition VI-1-1 in \cite{neveu1972martingales}, there exists $(\tilde{\nu}^k)_{k\in \N}\subset \mcA_{\tau,T}$ such that $J(\tau, \tilde{\nu}^k) \uparrow V_\tau$. By defining $\nu^k_s = \nu_s\, \ch_{t\leq s< \tau} + \tilde{\nu}^k\, \ch_{\tau\leq s\leq T}$, we obtain
    \begin{align}
        \frac{K_T^{\nu^k}}{K_\tau^{\nu^k}} &= \frac{K_T^{\tilde{\nu}^k}}{K_\tau^{\tilde{\nu}^k}} \quad \text{and}\quad
        Q_{\tau,T}^{\nu^k} = Q_{\tau,T}^{\tilde{\nu}^k}\,.
    \end{align}
    Thus, it follows that
    \begin{equation}
        \begin{split}
            \E^{\nu^k}_\tau \left[ - \exp \left\{ -\gamma\, \left(\rew + \rwealth_{\tau,T}^{\nu^k}\right) \right\} \right] &= \E_\tau \left[ - \frac{K_T^{\nu^k}}{K_\tau^{\nu^k}}\, \exp \left\{ -\gamma\, \left(\rew + \rwealth_{\tau,T}^{\nu^k}\right) \right\} \right]\\
            &= \E_\tau \left[ - \frac{K_T^{\tilde{\nu}^k}}{K_\tau^{\tilde{\nu}^k}}\, \exp \left\{ -\gamma\, \left(\rew + \rwealth_{\tau,T}^{\tilde{\nu}^k}\right) \right\} \right]\\
            &= J(\tau, \tilde{\nu}^k)\,.
        \end{split}
    \end{equation}
    We then obtain
    \begin{equation}
        \begin{split}
            J(t,\nu^k) &= \E_t \left[ - \frac{K_T^{\nu^k}}{K_\tau^{\nu^k}}\, \frac{K_\tau^{\nu^k}}{K_t^{\nu^k}}\, \exp \left\{ -\gamma\, \rwealth_{t,\tau}^{\nu^k} \right\}\, \exp \left\{ -\gamma\, \left(\rew + \rwealth_{\tau,T}^{\nu^k}\right) \right\} \right]\\
            &= \E_t \left[ \frac{K_\tau^{\nu^k}}{K_t^{\nu^k}}\, \exp \left\{ -\gamma\, \rwealth_{t,\tau}^{\nu^k} \right\}\, \E_\tau \left[- \frac{K_T^{\nu^k}}{K_\tau^{\nu^k}}\, \exp \left\{ -\gamma\, \left(\rew + \rwealth_{\tau,T}^{\nu^k}\right) \right\} \right]\right]\\
            &= \E_t \left[ \frac{K_\tau^{\nu}}{K_t^{\nu}}\, \exp \left\{ -\gamma\, \rwealth_{t,\tau}^{\nu} \right\}\, \E^{\nu^k}_\tau \left[  - \exp \left\{ -\gamma\, \left(\rew + \rwealth_{\tau,T}^{{\nu}^k}\right) \right\} \right] \right]\\
            &= \E_t \left[ \frac{K_\tau^{\nu}}{K_t^{\nu}}\, \exp \left\{ -\gamma\, \rwealth_{t,\tau}^{\nu} \right\}\, J(\tau, \tilde{\nu}^k) \right]\,.
        \end{split}
    \end{equation}
    Note that ${K_T^{\nu}}/{K_\tau^{\nu}}$ has mean of unity and is independent of $\F_\tau$. Thus,
    \begin{equation}
        \begin{split}
            \E_t \left[ \frac{K_\tau^{\nu}}{K_t^{\nu}}\, \exp \left\{ -\gamma\, \rwealth_{t,\tau}^{\nu} \right\}\, J(\tau, \tilde{\nu}^k) \right]
            &= \E_t \left[ \frac{K_\tau^{\nu}}{K_t^{\nu}}\, \exp \left\{ -\gamma\, \rwealth_{t,\tau}^{\nu} \right\}\, J(\tau, \tilde{\nu}^k) \right] \E_t \left[\frac{K_T^{\nu}}{K_\tau^{\nu}} \right]\\
            &= \E_t \left[ \frac{K_T^{\nu}}{K_t^{\nu}}\, \exp \left\{ -\gamma\, \rwealth_{t,\tau}^{\nu} \right\}\, J(\tau, \tilde{\nu}^k) \right]\\
            &= \E^{\nu}_t \left[ \exp \left\{ -\gamma\, \rwealth_{t,\tau}^{\nu} \right\}\, J(\tau, \tilde{\nu}^k) \right]
        \end{split}
    \end{equation}
   Due to Condition \ref{assumptionRgamma}, the expression $\E^{\nu}_t \left[ \exp \left\{ -\gamma\, \rwealth_{t,\tau}^{\nu} \right\}\, J(\tau, \tilde{\nu}^1) \right]$ is finite and we can use monotone convergence theorem to obtain
    \begin{equation}
        \begin{split}
            \lim_{k\uparrow \infty} J(t, \nu^k) &=  \E^{\nu}_t \left[ \exp \left\{ -\gamma\, \rwealth_{t,\tau}^{\nu} \right\}\, \lim_{k\to \infty} J(\tau, \tilde{\nu}^k) \right]\\
            &= \E^{\nu}_t \left[ \exp \left\{ -\gamma\, \rwealth_{t,\tau}^{\nu} \right\}\, V_\tau \right]\,.
        \end{split}
    \end{equation}
    Thus, we obtain 
    \begin{equation}
        \begin{split}
            V_t &\geq \lim_{k\uparrow \infty} J(t, \nu^k)\\
            &= \E^{\nu}_t \left[ \exp \left\{ -\gamma\, \rwealth_{t,\tau}^{\nu} \right\}\, V_\tau \right]\,,
        \end{split}
    \end{equation}
    which concludes the proof.
\end{proof}

\begin{lemma}\label{boundproba}
    Let $\Gamma \in \mathcal F$ such that $\mathbb P(\Gamma) >0$. Then there exists $\varepsilon >0$ such that
    $$\mathbb P^\nu (\Gamma) >\varepsilon \qquad \forall \nu \in \mathcal A.$$
\end{lemma}

\begin{proof}
    We will prove this by contradiction. Assume there exists $\left(\nu^k\right)_{k\in \N}$ a sequence in $\mcA$ such that $\lim_{n\to \infty} \Pb^{\nu^n}(\Gamma) = 0$. Thus,
    \begin{align}
        0 &= \lim_{n\to \infty} \E^{\nu^n}\left[ \ch_{\Gamma} \right]\\
        &=  \lim_{n\to \infty} \E \left[ \big|K_T^{\nu^n} \ch_{\Gamma}\big| \right]\,.
    \end{align}
    Then we can take a subsequence $(n_k)_{k\in \N} \subseteq (n)_{n\in \N}$ such that $\lim_{k\to \infty} K_T^{\nu^{n_k}} \ch_{\Gamma} = 0$ $\Pb$-almost surely. Because $\Gamma$ has $\Pb$-positive probability, then $\Pb (\lim_{k\to \infty} K_T^{\nu^{n_k}} = 0) > 0$, or equivalently,
    \begin{align}
        \Pb \left( \lim_{k\to \infty} \int_0^T \frac{\nu_t^{n_k}}{\eta}\, \d B_t - \frac{1}{2}\int_0^T \frac{(\nu_t^{n_k})^2}{\eta^2}\, \d t = -\infty \right) >0\,,\\
        \Leftrightarrow  \Pb \left( \lim_{k\to \infty} \int_0^T \frac{\nu_t^{n_k}}{\eta}\, \d B_t = -\infty \right) >0\,. \label{eq: contradiction}
    \end{align}
    However, by Ito isometry,
    \begin{align}
        \E \left[\left(\int_0^T \frac{\nu_u^{n_k}}{\eta}\, \d B_u\right)^2 \right] &= \E \left[\int_0^T \frac{(\nu_u^{n_k})^2}{\eta^2}\, \d u\right]\\
        &\leq \frac{T\, (\nu_\infty)^2}{\eta^2}\,,
    \end{align}
    because $\nu^{n_k}$ is bounded by $\nu_\infty$ $\d \Pb \otimes \d t \text{ a.e.}$ given that  $\nu_t^{n_k} \in \mcA$.
    Thus, the collection of random variables $\left\{ \int_0^T \frac{\nu_u^{n_k}}{\eta}\, \d B_u \right\}_{k\in \N}$ is bounded in $L^2(\Omega, \F_T, \Pb)$. From Equation \eqref{eq: contradiction}, we have
    \begin{equation}
        \begin{split}
            \Pb \left(\liminf_{k\to \infty} \left| \int_0^T \frac{\nu_t^{n_k}}{\eta}\, \d B_t \right|^2 = \infty \right) > 0\\
            \Rightarrow \E \left[ \liminf_{k\to \infty} \left| \int_0^T \frac{\nu_t^{n_k}}{\eta}\, \d B_t \right|^2 \right] = \infty\,.
        \end{split}
    \end{equation}
    However, by Fatou's lemma,
    \begin{equation}
        \begin{split}
            \infty &= \E \left[ \liminf_{k\to \infty} \left| \int_0^T \frac{\nu_t^{n_k}}{\eta}\, \d B_t \right|^2 \right]\\
            &\leq \liminf_{k\to \infty}  \E \left[ \left| \int_0^T \frac{\nu_t^{n_k}}{\eta}\, \d B_t \right|^2 \right]\\
            &\leq \frac{T\, (\nu_\infty)^2}{\eta^2}\,,
        \end{split}
    \end{equation}
    which is a contradiction.
\end{proof}

\begin{lemma}\label{lem:posV}
    For all $t\in \mfT$, we have $V_t < 0$ $\Pb$-almost surely.
\end{lemma}

\begin{proof}
    Let $\Gamma = \{ V_t = 0 \}$. By \cite{neveu1972martingales}, there exists $\left(\nu^k\right)_{k \in \N}$ a sequence in $\mcA_{t,T}$ such that $J(t,\nu^k)\uparrow 0$ on $\Gamma$. From monotone convergence theorem, $\E\left[ \ch_{\Gamma}\, J(t,\nu^k) \right]\to 0$. Then, there is a subsequence $k_1, k_2 , \cdots $ such that
    \begin{equation}
        K_T^{\nu^{k_m}}\,  \exp \left\{ -\gamma\, \left( \rew + \rwealth_{0,T}^{\nu^{k_m}} \right) \right\} \to 0 \quad \Pb-\text{a.s. on $\Gamma$}\,.
    \end{equation}
    This implies
    \begin{equation}
        \begin{split}
            &\int_0^T \frac{\nu^{k_m}_u}{\eta}\, \d B_u + \int_0^T  - \gamma\, \eta (S_{u\shortminus} + Z_{u\shortminus})\, \d B_u + \int_0^T \frac{\gamma\, \eta\, \nu^{k_m}_u\, (S_{u\shortminus} + Z_{u\shortminus})}{\eta}\, \d s\\
            - &\int_0^T  \gamma\, \sigma\, Y_{u\shortminus}\, \d W_u -  \sum_{i\in\{+,-\}} \int_0^T \gamma\, \left\{ \deltai\, \xi\, \left( S_{u} - \frac{Z_{u\shortminus}\, Y_{u\shortminus}}{Y_{u\shortminus} + \deltai\, \xi} \right) \right\}\, \d \Ni_u \to -\infty\, 
        \end{split}
    \end{equation}
$\Pb$-almost surely on $\Gamma$. However, 
\begin{align}
    \sup_{t\in \mfT} Y_t \leq Y_0 + \xi\, (\hatNa_T + \hatNb_T) + \sup_{t\in \mfT} B_t\,,\qquad \text{and} \qquad 
    \sup_{t\in \mfT} Z_t \leq Z_0\, 4^{\hatNa_T + \hatNb_t}\,.
\end{align}
    
Moreover, given that $\nu^{k_m}$ is bounded and using the above inequalities, we conclude that the last four terms in the above expression are finite a.s., therefore,  it has to be the case that
    \begin{equation}
        \int_0^T \frac{\nu^{k_m}_u}{\eta}\, \d B_u \to -\infty\qquad \Pb-\text{a.s. on $\Gamma$}\,.
    \end{equation}
From Lemma \ref{boundproba}, we should have $\Pb(\Gamma)=0$. Thus, we have $V_t < 0$ $\Pb$-almost surely.
\end{proof}

\subsection{Contract representation}\label{appendix 2}

\begin{proof}[\textbf{Proof of Theorem \ref{contractrep}}]
By the DPP in Lemma \ref{DPPMM}, one can show that for any $t\ge s$
$$V_s = \sup_{\nu \in \mcA} \E_s^\nu \left[\, \exp \left\{ -\gamma\, \rwealth^\nu_{s,t} \right\}\, V_t\,  \right]\,.$$
Then, $U^\nu$ defined as
$$U_t^\nu = \exp \left\{ -\gamma\, \rwealth^\nu_{0,t} \right\}\, V_t$$
is a $\Pb^\nu$-supermartingale. By Doob-Meyer decomposition, we can write $U^\nu$ as
\begin{equation}
    \d U^\nu_t = \d M^\nu_t - \d A^{\nu,c}_t - \d A^{\nu, d}_t\,,
\end{equation}
where $M^\nu$ is a $(\Fb, \Pb^\nu)$-martingale, $A^{\nu,c}$ is a continuous non-decreasing process, $A^{\nu, d}$ is a pure-jump non-decreasing process, and $A^{\nu,c}_0 = A^{\nu,d}_0 = 0$. By Martingale Representation Theorem, there exists a $\Fb$-predictable process $(\tilde{A}^{\nu, W}, \tilde{A}^{\nu, B}, \tilde{A}^{\nu, -}, \tilde{A}^{\nu, +})$ such that
\begin{equation}
    \d M^\nu_t = \tilde{A}^{\nu, W}_t\, \d W_t + \tilde{A}^{\nu, B}_t\, \d B^\nu_t + \sum_{i\in\{+,-\}} \tilde{A}^{\nu, i}_t\, \d \tildeNi_t\,.
\end{equation}
By Ito's formula, we have that
\begin{equation}
    \begin{split}
        \d V_t &= U^\nu_{t\shortminus}\, \d \left( \exp \left\{ \gamma\, \rwealth^\nu_{0,t} \right\} \right) + \exp \left\{ \gamma\, \rwealth^\nu_{0,t\shortminus} \right\}\, \d\, U^\nu_t\\
        &\phantom{{}={}}+ \sum_{i\in\{+,-\}} (U^\nu_{t} - U^\nu_{t\shortminus})\left(\exp \left\{ \gamma\, \rwealth^\nu_{0,t} \right\} - \exp \left\{ \gamma\, \rwealth^\nu_{0,t\shortminus} \right\} \right)\, \d \hatNi_t\\
        &\phantom{{}={}} + \d \langle U^{\nu, c}, \exp \left\{ \gamma\, \rwealth^{\nu,c}_{0,\cdot} \right\}\rangle_t
    \end{split}
\end{equation}
where $U^{\nu, c}$ and $\rwealth^{\nu,c}_{0,\cdot}$ are the continuous parts of $U^{\nu}$ and $\rwealth^{\nu}_{0,\cdot}$ respectively. As $\langle \hat \Ni, A^{\nu,d}\rangle = 0$, we have
$$ (U^\nu_{t} - U^\nu_{t\shortminus})\left(\exp \left\{ \gamma\, \rwealth^\nu_{0,t} \right\} - \exp \left\{ \gamma\, \rwealth^\nu_{0,t\shortminus} \right\} \right)\, \d \hatNi_t = \tilde{A}^{\nu,i}_t \left(\exp \left\{ \gamma\, \rwealth^\nu_{0,t} \right\} - \exp \left\{ \gamma\, \rwealth^\nu_{0,t\shortminus} \right\} \right)\, \d \hatNi_t\,.$$
We also calculate
\begin{align}
    \d \left( \exp \left\{ \gamma\, \rwealth^\nu_{0,t} \right\} \right) &= \gamma\, \left( \exp \left\{ \gamma\, \rwealth^\nu_{0,t\shortminus} \right\} \right)\, \d \rwealth^{\nu,c}_{0,t} + \frac{\gamma^2}{2}\, \left( \exp \left\{ \gamma\, \rwealth^\nu_{0,t\shortminus} \right\} \right)\, \d \langle \rwealth^{\nu,c}_{0,\cdot} \rangle_t\\
    &\phantom{{}={}} + \sum_{i\in\{+,-\}}\, \left[ \exp \left\{ \gamma\, \rwealth^\nu_{0,t} \right\} - \exp \left\{ \gamma\, \rwealth^\nu_{0,t\shortminus} \right\}\right]\, \d \hatNi_t\,,\footnotemark
\end{align}\footnotetext{Note that even though the jump processes included in $\rwealth_{0,\cdot}^\nu$ are $\Na$ and not $\hatNa$, we have $\rwealth^\nu_{0,t}= \rwealth^\nu_{0,t\shortminus}$ and $\Deltarwealtha_t = 0$ on the event $\{ \Na_t = \Na_{t\shortminus}+1, Y_{t\shortminus}\leq \xi\}$.} 
and finally, 
\begin{align}
    \d \langle U^{\nu, c}, \exp \left\{ \gamma\, \rwealth^{\nu,c}_{0,\cdot} \right\}\rangle_t &= \gamma\, \exp \left\{ \gamma\, \rwealth^\nu_{0,t\shortminus} \right\}\, \left(\sigma\, Y_{t\shortminus}\, \tilde{A}^{\nu, W}_t + \eta\, (S_{t\shortminus} + Z_{t\shortminus})\, \tilde{A}^{\nu\, B}_t\right)\, \d t\\
    &= \frac{\gamma\, V_{t\shortminus}}{U^{\nu}_{t\shortminus}}\, \left(\sigma\, Y_{t\shortminus}\, \tilde{A}^{\nu, W}_t + \eta\, (S_{t\shortminus} + Z_{t\shortminus})\, \tilde{A}^{\nu\, B}_t\right)\, \d t\,.
\end{align}
Thus, we are able to rewrite $\d V_t$ as
\begin{align}
    \d V_t &= \gamma\, U^\nu_{t\shortminus}\, \left( \exp \left\{ \gamma\, \rwealth^\nu_{0,t\shortminus} \right\} \right)\, \d \rwealth^{\nu,c}_{0,t} + \frac{\gamma^2\, U^\nu_{t\shortminus}}{2}\, \left( \exp \left\{ \gamma\, \rwealth^\nu_{0,\shortminus} \right\} \right)\, \d \langle \rwealth^{\nu,c}_{0,\cdot} \rangle_t\\
    &\phantom{{}={}} + \sum_{i\in\{+,-\}}\, U^\nu_{t\shortminus}\, \left[ \exp \left\{ \gamma\, \rwealth^\nu_{0,t} \right\} - \exp \left\{ \gamma\, \rwealth^\nu_{0,t\shortminus} \right\}\right]\, \d \hat {N}^i_t\\
    &\phantom{{}={}} + \exp \left\{ \gamma\, \rwealth^\nu_{0,t\shortminus} \right\}\, \tilde{A}^{\nu, W}_t\, \d W_t + \exp \left\{ \gamma\, \rwealth^\nu_{0,t\shortminus} \right\}\, \tilde{A}^{\nu, B}_t\, \d B^\nu_t + \sum_{i\in\{+,-\}} \exp \left\{ \gamma\, \rwealth^\nu_{0,t\shortminus} \right\}\, \tilde{A}^{\nu, i}_t\, \left(\d \hatNi_t - \lambdai_t\, \d t\right)\\
    &\phantom{{}={}} - \exp \left\{ \gamma\, \rwealth^\nu_{0,t\shortminus} \right\}\, \d A^{\nu,c}_t - \exp \left\{ \gamma\, \rwealth^\nu_{0,t\shortminus} \right\}\, \d A^{\nu, d}_t + \sum_{i\in\{+,-\}} \tilde{A}^{\nu,i}_t\,\left(\exp \left\{ \gamma\, \rwealth^\nu_{0,t} \right\} - \exp \left\{ \gamma\, \rwealth^\nu_{0,t\shortminus} \right\} \right)\, \d \hat {N}^i_t\\
    &\phantom{{}={}} + \frac{\gamma\, V_{t\shortminus}}{U^{\nu}_{t\shortminus}}\, \left(\sigma\, Y_{t\shortminus}\, \tilde{A}^{\nu, W}_t + \eta\, (S_{t\shortminus} + Z_{t\shortminus})\, \tilde{A}^{\nu\, B}_t\right)\, \d t\,,\\
    &= \gamma\, V_{t\shortminus}\, \d \rwealth^{\nu,c}_{0,t} + \frac{\gamma^2\, V_{t\shortminus}}{2}\,\d \langle \rwealth^{\nu,c}_{0,\cdot} \rangle_t\\
    &\phantom{{}={}} + \sum_{i\in\{+,-\}}\, V_{t\shortminus}\, \left( \exp \left\{ \gamma\, \Deltarwealthi_{t} \right\} - 1\right)\, \d \hat {N}^i_t\\
    &\phantom{{}={}} + \frac{V_{t\shortminus}\, \tilde{A}^{\nu, W}_t}{U^\nu_{t\shortminus}}\, \d W_t + \frac{V_{t\shortminus}\, \tilde{A}^{\nu, B}_t}{U^\nu_{t\shortminus}}\, \d B^\nu_t + \sum_{i\in\{+,-\}} \frac{V_{t\shortminus}\, \tilde{A}^{\nu, i}_t}{U^\nu_{t\shortminus}}\, \left(\d \hatNi_t - \lambdai_t\, \d t\right)\\
    &\phantom{{}={}} - \frac{V_{t\shortminus}}{U^\nu_{t\shortminus}}\, \d A^{\nu,c}_t - \frac{V_{t\shortminus}}{U^\nu_{t\shortminus}}\, \d A^{\nu, d}_t + \sum_{i\in\{+,-\}} \frac{V_{t\shortminus}\, \tilde{A}^{\nu, i}_t}{U^\nu_{t\shortminus}}\,\left(\exp \left\{ \gamma\, \Deltarwealthi_{t} \right\} - 1 \right)\, \d \hat {N}^i_t\\
    &\phantom{{}={}} + \frac{\gamma\, V_{t\shortminus}}{U^{\nu}_{t\shortminus}}\, \left(\sigma\, Y_{t\shortminus}\, \tilde{A}^{\nu, W}_t + \eta\, (S_{t\shortminus} + Z_{t\shortminus})\, \tilde{A}^{\nu\, B}_t\right)\, \d t\,.
\end{align}
Let $P_t = -\frac{1}{\gamma} \log (-V_t)$. Then, by Ito's formula
\begin{align}
    \d P_t &= -\frac{1}{\gamma\, V_{t\shortminus}}\, \d V_t^{c} + \frac{1}{2\,\gamma\, V_{t\shortminus}^2}\, \d \langle V^c \rangle_t - \frac{1}{\gamma}\, \sum_{i\in\{+,-\}} \log \left(\frac{V_t}{V_{t\shortminus}} \right)\, \d \hat \Ni_t - \frac{1}{\gamma\, \Delta A_{t}^{\nu, d}}\, \log \left(\frac{V_t}{V_{t\shortminus}} \right)\, \d A^{\nu, d}_t\,\\
    &= -\frac{1}{\gamma\, V_{t\shortminus}}\, \d V_t^{c} + \frac{1}{2\,\gamma\, V_{t\shortminus}^2}\, \d \langle V^c \rangle_t - \frac{1}{\gamma}\, \sum_{i\in\{+,-\}} \log \left( 1 +\frac{\Delta^i V_t}{V_{t\shortminus}} \right)\, \d \hat \Ni_t - \frac{1}{\gamma\, \Delta A_{t}^{\nu, d}}\, \log \left(1 + \frac{\Delta^d V_t}{V_{t\shortminus}} \right)\, \d A^{\nu, d}_t\,,
\end{align}
with $V^c$ is the continuous part of $V$ and $Y_T = \rew$. We calculate
\begin{align}
    \frac{d\langle V^c\rangle_t}{\d t} &= \left(\frac{V_{t\shortminus}\, \tilde{A}^{\nu, W}_t}{U^\nu_{t\shortminus}} + \gamma\, V_{t\shortminus}\, \sigma\, Y_{t\shortminus}\right)^2 + \left(\frac{V_{t\shortminus}\, \tilde{A}^{\nu, B}_t}{U^\nu_{t\shortminus}} + \gamma\, V_{t\shortminus}\, \eta\, (S_t + Z_{t\shortminus})\right)^2 \,,\\
   \Delta^i V_t &= \left(1 + \frac{\tilde{A}^{\nu,i}_t}{U^\nu_{t\shortminus}}\right)\, V_{t\shortminus}\, \exp \left\{ \gamma\, \Deltarwealthi_{t}  \right\} - V_{t\shortminus}\,,\\
   \Delta^d V_t &=  - \frac{V_{t\shortminus}\, \Delta {A}^{\nu,d}_t}{U^\nu_{t\shortminus}}\,,\\
   \d V^c_t &= \gamma\, V_{t\shortminus}\, \d \rwealth^{\nu,c}_{0,t} + \frac{\gamma^2\, V_{t\shortminus}}{2}\,\d \langle \rwealth^{\nu,c}_{0,\cdot} \rangle_t + \frac{\gamma\, V_{t\shortminus}}{U^{\nu}_{t\shortminus}}\, \left(\sigma\, Y_{t\shortminus}\, \tilde{A}^{\nu, W}_t + \eta\, (S_{t\shortminus} + Z_{t\shortminus})\, \tilde{A}^{\nu\, B}_t\right)\, \d t\\
   &\phantom{{}={}} +\frac{V_{t\shortminus}\, \tilde{A}^{\nu, W}_t}{U^\nu_{t\shortminus}}\, \d W_t + \frac{V_{t\shortminus}\, \tilde{A}^{\nu, B}_t}{U^\nu_{t\shortminus}}\, \d B^\nu_t - \sum_{i\in\{+,-\}} \frac{\lambdai_t\, V_{t\shortminus}\, \tilde{A}^{\nu, i}_t}{U^\nu_{t\shortminus}}\, \d t - \frac{V_{t\shortminus}}{U^\nu_{t\shortminus}}\, \d A^{\nu,c}_t\,.
\end{align}
We may rewrite $\d P_t$ as
\begin{align}
    \d P_t &= -\d \rwealth^{\nu,c}_{0,t} - \frac{\gamma}{2}\,\d \langle \rwealth^{\nu,c}_{0,\cdot} \rangle_t - \frac{\tilde{A}^{\nu, W}_t}{\gamma\, U^\nu_{t\shortminus}}\, \d W_t - \frac{\tilde{A}^{\nu, B}_t}{\gamma\, U^\nu_{t\shortminus}}\, \d B^\nu_t + \sum_{i\in\{+,-\}} \frac{\lambdai_t\, \tilde{A}^{\nu, i}_t}{\gamma\, U^\nu_{t\shortminus}}\, \d t + \frac{1}{\gamma\, U^\nu_{t\shortminus}}\, \d A^{\nu,c}_t\\
    &\phantom{{}={}} + \frac{1}{2\, \gamma}\, \left[ \left(\frac{\tilde{A}^{\nu, W}_t}{U^\nu_{t\shortminus}} + \gamma\, \sigma\, Y_{t\shortminus}\right)^2 + \left(\frac{\tilde{A}^{\nu, B}_t}{U^\nu_{t\shortminus}} + \gamma\, \eta\, (S_t + Z_{t\shortminus})\right)^2 \right]\, \d t\\
    &\phantom{{}={}} - \frac{1}{U^{\nu}_{t\shortminus}}\, \left(\sigma\, Y_{t\shortminus}\, \tilde{A}^{\nu, W}_t + \eta\, (S_{t\shortminus} + Z_{t\shortminus})\, \tilde{A}^{\nu\, B}_t\right)\, \d t\\
    &\phantom{{}={}} - \frac{1}{\gamma}\, \sum_{i\in\{+,-\}} \left\{\log \left( 1 + \frac{\tilde{A}^{\nu,i}_t}{U^\nu_{t\shortminus}} \right) + \gamma\, \Deltarwealthi_{t} \right\}\, \d \hat {N}^i_t\\
    &\phantom{{}={}} - \frac{1}{\gamma\, \Delta A_{t}^{\nu, d}}\, \log \left(1 - \frac{\Delta A^{\nu,d}_t}{U^\nu_{t\shortminus}} \right)\, \d A^{\nu, d}_t\\
    &= \Bigg\{ \tempLP\, \nu_{t\shortminus}^2  + \frac{1}{2\, \gamma}\, \left[ \left(\frac{\tilde{A}^{\nu, W}_t}{U^\nu_{t\shortminus}} + \gamma\, \sigma\, Y_{t\shortminus}\right)^2 + \left(\frac{\tilde{A}^{\nu, B}_t}{U^\nu_{t\shortminus}} + \gamma\, \eta\, (S_t + Z_{t\shortminus})\right)^2 \right]\\
    &\phantom{{}={}} - \frac{\gamma}{2}\, \left[ (\sigma\,Y_{t\shortminus})^2 + \eta^2\, (S_t + Z_{t\shortminus})^2 \right] + \sum_{i\in\{+,-\}} \frac{\lambdai_t\, \tilde{A}^{\nu, i}_t}{\gamma\, U^\nu_{t\shortminus}} + \frac{1}{\gamma\, U^\nu_{t\shortminus}}\, \frac{\d A^{\nu,c}_t}{\d t}\\
    &\phantom{{}={}} - \frac{1}{U^{\nu}_{t\shortminus}}\, \left(\sigma\, Y_{t\shortminus}\, \tilde{A}^{\nu, W}_t + \eta\, (S_{t\shortminus} + Z_{t\shortminus})\, \tilde{A}^{\nu\, B}_t\right)\Bigg\}\, \d t\\
    &\phantom{{}={}} - \frac{1}{\gamma}\, \sum_{i\in\{+,-\}} \left\{\log \left( 1 + \frac{\tilde{A}^{\nu,i}_t}{U^\nu_{t\shortminus}} \right) + \gamma\, \Deltarwealthi_{t} \right\}\, \d \hat {N}^i_t\\
    &\phantom{{}={}} - \frac{1}{\gamma\, \Delta A_{t}^{\nu, d}}\, \log \left(1 - \frac{\Delta A^{\nu,d}_t}{U^\nu_{t\shortminus}} \right)\, \d A^{\nu, d}_t\\
    &\phantom{{}={}} - \left\{ \sigma\, Y_{t\shortminus} + \frac{\tilde{A}^{\nu, W}_t}{\gamma\, U^\nu_{t\shortminus}} \right\}\, \d W_t - \left\{ \eta\, (S_{t\shortminus} + Z_{t\shortminus}) + \frac{\tilde{A}^{\nu, B}_t}{\gamma\, U^\nu_{t\shortminus}} \right\}\, \d B^\nu_t\\
    &= \Bigg\{ \tempLP\, \nu_{t\shortminus}^2 + \frac{1}{2\, \gamma}\, \left[ \left(\frac{\tilde{A}^{\nu, W}_t}{U^\nu_{t\shortminus}} + \gamma\, \sigma\, Y_{t\shortminus}\right)^2 + \left(\frac{\tilde{A}^{\nu, B}_t}{U^\nu_{t\shortminus}} + \gamma\, \eta\, (S_t + Z_{t\shortminus})\right)^2 \right]\\
    &\phantom{{}={}} - \frac{\gamma}{2}\, \left[ (\sigma\,Y_{t\shortminus})^2 + \eta^2\, (S_t + Z_{t\shortminus})^2 \right] + \sum_{i\in\{+,-\}} \frac{\lambdai_t\, \tilde{A}^{\nu, i}_t}{\gamma\, U^\nu_{t\shortminus}} + \frac{1}{\gamma\, U^\nu_{t\shortminus}}\, \frac{\d A^{\nu,c}_t}{\d t}\\
    &\phantom{{}={}} - \frac{1}{U^{\nu}_{t\shortminus}}\, \left(\sigma\, Y_{t\shortminus}\, \tilde{A}^{\nu, W}_t + \eta\, (S_{t\shortminus} + Z_{t\shortminus})\, \tilde{A}^{\nu\, B}_t\right) - \nu_{t\shortminus}\, (S_{t\shortminus} + Z_{t\shortminus}) - \frac{\tilde{A}^{\nu, B}_t\, \nu_{t\shortminus}}{\gamma\, \eta\, U^\nu_{t\shortminus}}\Bigg\}\, \d t\\
    &\phantom{{}={}} - \frac{1}{\gamma}\, \sum_{i\in\{+,-\}} \left\{\log \left( 1 + \frac{\tilde{A}^{\nu,i}_t}{U^\nu_{t\shortminus}} \right) + \gamma\, \Deltarwealthi_{t} \right\}\, \d \hat {N}^i_t\\
    &\phantom{{}={}} - \frac{1}{\gamma\, \Delta A_{t}^{\nu, d}}\, \log \left(1 - \frac{\Delta A^{\nu,d}_t}{U^\nu_{t\shortminus}} \right)\, \d A^{\nu, d}_t\\
    &\phantom{{}={}} - \left\{ \sigma\, Y_{t\shortminus} + \frac{\tilde{A}^{\nu, W}_t}{\gamma\, U^\nu_{t\shortminus}} \right\}\, \d W_t - \left\{ \eta\, (S_{t\shortminus} + Z_{t\shortminus}) + \frac{\tilde{A}^{\nu, B}_t}{\gamma\, U^\nu_{t\shortminus}} \right\}\, \d B_t\,.
\end{align}

Simplifying the above equations, we get
\begin{align}
    \d P_t &= \Bigg\{ \tempLP\, \nu_{t\shortminus}^2 + \frac{1}{2\, \gamma}\, \left[ \left(\frac{\tilde{A}^{\nu, W}_t}{U^\nu_{t\shortminus}}\right)^2 + \left(\frac{\tilde{A}^{\nu, B}_t}{U^\nu_{t\shortminus}}\right)^2 \right]\\
    &\phantom{{}={}} + \sum_{i\in\{+,-\}} \frac{\lambdai_t\, \tilde{A}^{\nu, i}_t}{\gamma\, U^\nu_{t\shortminus}} + \frac{1}{\gamma\, U^\nu_{t\shortminus}}\, \frac{\d A^{\nu,c}_t}{\d t} + \nu_{t\shortminus}\, (S_{t\shortminus} + Z_{t\shortminus}) + \frac{\tilde{A}^{\nu, B}_t\, \nu_{t\shortminus}}{\gamma\, \eta\, U^\nu_{t\shortminus}}\Bigg\}\, \d t\\
    &\phantom{{}={}} - \frac{1}{\gamma}\, \sum_{i\in\{+,-\}} \left\{\log \left( 1 + \frac{\tilde{A}^{\nu,i}_t}{U^\nu_{t\shortminus}} \right) + \gamma\, \Deltarwealthi_{t} \right\}\, \d \hat {N}^i_t\\
    &\phantom{{}={}} - \frac{1}{\gamma\, \Delta A_{t}^{\nu, d}}\, \log \left(1 - \frac{\Delta A^{\nu,d}_t}{U^\nu_{t\shortminus}} \right)\, \d A^{\nu, d}_t\\
    &\phantom{{}={}} - \left\{ \sigma\, Y_{t\shortminus} + \frac{\tilde{A}^{\nu, W}_t}{\gamma\, U^\nu_{t\shortminus}} \right\}\, \d W_t - \left\{ \eta\, (S_{t\shortminus} + Z_{t\shortminus}) + \frac{\tilde{A}^{\nu, B}_t}{\gamma\, U^\nu_{t\shortminus}} \right\}\, \d B_t\,.
\end{align}

If we rewrite 
\begin{align}
    \Ai_t &= - \frac{1}{\gamma}\, \left(\log \left( 1 + \frac{\tilde{A}^{\nu,i}_t}{U^\nu_{t\shortminus}} \right) + \gamma\, \Deltarwealthi_{t} \right)\\
    \AW_t &= - \left( \sigma\, Y_{t\shortminus} + \frac{\tilde{A}^{\nu, W}_t}{\gamma\, U^\nu_{t\shortminus}} \right)\\
    \AB_t &= - \left( \eta\, (S_{t\shortminus} + Z_{t\shortminus}) + \frac{\tilde{A}^{\nu, B}_t}{\gamma\, U^\nu_{t\shortminus}} \right)\,
\end{align}
where $i\in\{+,-\}$, we  have
\begin{align}
    \d P_t &= \Bigg\{ \tempLP\, \nu_{t\shortminus}^2 + \frac{\gamma}{2}\, \left[ \left(\AW_t + \sigma\, Y_{t\shortminus}\right)^2 + \left(\AB_t + \eta\, (S_{t\shortminus} + Z_{t\shortminus})\right)^2 \right]\\
    &\phantom{{}={}} - \sum_{i\in\{+,-\}} \frac{\lambdai_t\,(1 - e^{-\gamma(\Ai_t + \Deltarwealthi_t)})}{\gamma} + \frac{1}{\gamma\, U^\nu_{t\shortminus}}\, \frac{\d A^{\nu,c}_t}{\d t} - \frac{\AB_t\, \nu_{t\shortminus}}{\eta}\Bigg\}\, \d t\\
    &\phantom{{}={}} + \sum_{i\in\{+,-\}} \Ai_t\, \d \hat {N}^i_t - \frac{1}{\gamma\, \Delta A_{t}^{\nu, d}}\, \log \left(1 - \frac{\Delta A^{\nu,d}_t}{U^\nu_{t\shortminus}} \right)\, \d A^{\nu, d}_t\\
    &\phantom{{}={}} + \AW_t\, \d W_t + \AB_t\, \d B_t\,.
\end{align}
Note that
\begin{align}
    P_u &= \rew - \int_u^T  \Bigg\{ \tempLP\, \nu_{t\shortminus}^2 + \frac{\gamma}{2}\, \left[ \left(\AW_t + \sigma\, Y_{t\shortminus}\right)^2 + \left(\AB_t + \eta\, (S_{t\shortminus} + Z_{t\shortminus})\right)^2 \right]\\
    &\phantom{{}={}} - \sum_{i\in\{+,-\}} \frac{\lambdai_t\,(1 - e^{-\gamma(\Ai_t + \Deltarwealthi_t)})}{\gamma} + \frac{1}{\gamma\, U^\nu_{t\shortminus}}\, \frac{\d A^{\nu,c}_t}{\d t} - \frac{\AB_t\, \nu_{t\shortminus}}{\eta}\Bigg\}\, \d t\\
    &\phantom{{}={}} -\sum_{i\in\{+,-\}} \int_u^T  \Ai_t\, \d \hat {N}^i_t - \int_u^T \AW_t\, \d W_t -\int_u^T \AB_t\, \d B_t + \int_u^T \d A^d_t\,,
\end{align}
with  $\d A^d_t = \frac{1}{\gamma\, \Delta A_{t}^{\nu, d}}\, \log \left(1 - \frac{\Delta A^{\nu,d}_t}{U^\nu_{t\shortminus}} \right)\, \d A^{\nu, d}_t$, i.e.
$$A^d_t = \frac 1\gamma \sum_{s\le t}\log \left(1 - \frac{\Delta A^{\nu,d}_t}{U^\nu_{t\shortminus}} \right).$$
This shows in particular that $\Delta a_t = \frac{-\Delta A^{\nu, d}_t}{U^\nu_{t-}} \ge 0$ is independent from $\nu \in \mathcal A$.\\

Therefore, we have $P_T = \rew$ and
\begin{align}
    \d P_t &=  - \Bigg\{\bar h \left(\nu_t, A_t, Z_t, Y_t, S_t \right) -\frac{1}{\gamma\, U^\nu_{t\shortminus}}\, \frac{\d A^{\nu,c}_t}{\d t} \Bigg\}\, \d t +\sum_{i\in\{+,-\}}  \Ai_t\, \d \hat {N}^i_t +  \AW_t\, \d W_t + \AB_t\, \d B_t -  \d A^d_t\,
\end{align}
with
\begin{align}
 \bar h \left(\nu, A, Z, Y, S \right) & = h(\nu, A) + \sum_{i\in\{+,-\}} \frac{\barlambdai(Z, Y, S)\,(1 - e^{-\gamma(\Ai + \Deltarwealthi)})}{\gamma} - \frac{\gamma}{2}\, \left[ \left(\AW + \sigma\, Y\right)^2 + \left(\AB + \eta\, (S + Z)\right)^2 \right].   
\end{align}

We introduce the process $I$ given by
$$I_t = \int_0^t \Bigg\{\bar h \left(\nu_s, A_s, Z_s, Y_s, S_s \right) \d s-\frac{1}{\gamma\, U^\nu_{s\shortminus}}\, \d A^{\nu,c}_s \Bigg\}.$$
By the DPP, we have

\begin{align}
    0 &= \sup_{\nu \in \mathcal A} \mathbb E^\nu \left[ U^\nu_T \right] - V_0 = \sup_{\nu \in \mathcal A} \mathbb E^\nu \left[ U^\nu_T  - M^\nu_T \right] = \sup_{\nu \in \mathcal A} \mathbb E^\nu \left[\int_0^T \left( -\d A^{\nu, c}_t - \d A^{\nu, d}_t \right) \right]\\
    &= \gamma \sup_{\nu \in \mathcal A} \mathbb E^\nu \left[\int_0^T U^\nu_{t-}\left( \d I_t - \bar h \left(\nu_t, A_t, Z_t, Y_t, S_t \right) \d t  + \frac{\d a_t }{\gamma}\right) \right]
\end{align}
where $\d a_t = \frac{-\d A^{\nu, d}_t}{U^\nu_{t-}} \ge 0.$ \\

Moreover, by Lemma \ref{lem:posV}, we have
$$U^\nu_t = e^{-\gamma \rwealth^\nu_{0,t}} V_t \le -\beta_t V_t  <0,$$
where we write
$$\beta_t = -e^{-\gamma \left(  + \eta \int_0^t (S_u + Z_{u\shortminus}) \d B_u + \int_0^t |S_u + Z_{u\shortminus}|\nu_{\infty} \d u + \sigma \int_0^t Y_{t\shortminus} \d W_u + \underset{i=a,b }{\sum}\int_0^t \left\{ \deltai\, \xi\, \left( S_{u} - \frac{Z_{u\shortminus}\, Y_{u\shortminus}}{Y_{u\shortminus} + \deltai\, \xi} \right) \right\}\, \d \Ni_u \right) }.$$
Therefore we have
\begin{align}
    0 & \le \sup_{\nu \in \mathcal A} \mathbb E^\nu \left[\int_0^T -\beta_t V_t\left( \d I_t - \bar h \left(\nu_t, A_t, Z_t, Y_t, S_t \right) \d t  + \frac{\d a_t }{\gamma}\right) \right]\\
    & \le \sup_{\nu \in \mathcal A} \mathbb E^\nu \left[\int_0^T -\beta_t V_t\left( \d I_t - \bar H \left(A_t, Z_t, Y_t, S_t \right) \d t  + \frac{\d a_t }{\gamma}\right) \right].
\end{align}

Let us denote by $\Theta$ the random variable
$$\Theta = \int_0^T -\beta_t V_t\left( \d I_t - \bar H \left(A_t, Z_t, Y_t, S_t \right) \d t  + \frac{\d a_t }{\gamma}\right)\le 0.$$
We then have
$$0\le \sup_{\nu \in \mathcal A} \mathbb E^\nu \left[\Theta \right].$$

Let $\varepsilon >0$. For all $\nu \in \mathcal A$, we have 
$$\mathbb E^\nu \left[\Theta \right]  = \mathbb E^\nu \left[\Theta \mathds{1}_{\Theta \le -\varepsilon} + \Theta \mathds{1}_{\Theta >\varepsilon} \right] \le -\varepsilon \mathbb P^\nu (\Theta \le - \varepsilon ).$$
Therefore,
$$0\le \sup_{\nu \in \mathcal A} -\varepsilon \mathbb P^\nu (\Theta \le - \varepsilon ),$$
which implies that 
$$\inf_{\nu \in \mathcal A}  \mathbb P^\nu (\Theta \le - \varepsilon ) = 0.$$

By Lemma \ref{boundproba}, this implies that 
$$\mathbb P (\Theta \le - \varepsilon ) = 0.$$
This is true for all $\varepsilon >0$, therefore we have that 

$$\Theta = \int_0^T -\beta_t V_t\left( \d I_t - \bar H \left(A_t, Z_t, Y_t, S_t \right) \d t  + \frac{\d a_t }{\gamma}\right) = 0 \qquad \text{a.s.}$$

Since $\beta V > 0,$ $\d I_t - \bar H \left(A_t, Z_t, Y_t, S_t \right)\d t \ge 0 $, and $\d a_t \ge 0,$ this finally implies that 
$$\d I_t  = \bar H \left(A_t, Z_t, Y_t, S_t \right)\d t \quad \text{and} \quad \d a_t = \d A^d_t = 0.$$

It remains to prove that $A\in \Lambda$, i.e.
\begin{equation}
    \sup_{\nu \in \mcA} \sup_{t\in \mfT} \E^\nu \left[ \exp\left\{ -\gamma'\, P_t \right\}  \right] < \infty \quad \text{ for some } \gamma'>\gamma\,.
\end{equation}
Condition \ref{assumptionRgamma} together with Hölder inequality guarantees that there exists $r>0$ such that 
$$ \sup_{\nu \in \mcA}  \E^\nu \left[ | U^\nu_T|^{r+1} \right] < \infty.$$
Therefore, as $U^\nu$ is a negative $\mathbb P^\nu -$supermartingale, we get
$$\sup_{\nu \in \mcA} \sup_{t\in \mfT}  \E^\nu \left[ | U^\nu_t|^{r+1} \right] =  \sup_{\nu \in \mcA}  \E^\nu \left[ | U^\nu_T|^{r+1} \right] < \infty,$$
which leads to the result by using again Hölder inequality and the fact that $e^{-\gamma P_t} = U^\nu_t e^{\gamma Q^\nu_{0,t}}.$

\end{proof}

\begin{proof}[\textbf{Proof of Theorem \ref{theorem: contract rep value}}]
    Let $\rew = P_T^{P_0, A}$, $\bar{P}_t = P_t^{P_0, A} + \rwealth_{0,t}^\nu$. Then
    \begin{align}
        \d e^{-\gamma\, \bar{P}_t} &= \gamma\, e^{-\gamma\, \bar{P}_{t\shortminus}}\, \Bigg\{ \big(H(A_t) - h(\nu_t, A_t)\big)\, \d t - \AW_t\, \d W_t - \AB_t\, \d B_t^{\nu} + \sum_{i\in\{+,-\}} \frac{e^{-\gamma\, (\Deltarwealthi_t + \Ai_t)}-1}{\gamma}\, \d \tildeNi_t\Bigg\}\,.
    \end{align}

    We observe that $H(A)\geq h(\nu, A)$. As $A\in \Lambda$, the process $(-e^{-\gamma\, \bar{P}_t})_{t\in \mfT}$ is a $\Pb^\nu$-supermartingale. Thus,
    \[  \E^\nu \left[ -e^{-\gamma\, \bar{P}_T}\right] \leq -e^{-\gamma\, P_0}\,.  \]
    Equality applies if and only if $H(A_t) = h(\nu_t, A_t)$, $\d t \otimes \d \Pb$ a.s., i.e. $\nu_t = \bar{\nu}(A_t)$.
\end{proof}

\end{document}